\documentclass[11pt]{article}  

\usepackage{amsmath, amsthm, amssymb}
\usepackage{graphicx}

\usepackage[margin=1.2in]{geometry}

\usepackage[table]{xcolor}

\usepackage{soul}
\usepackage{times}

\usepackage{nameref,hyperref}

\usepackage{amsmath, amsthm, amssymb}
\usepackage{algorithm}
\usepackage[noend]{algpseudocode}

\newcommand{\argmax}{\operatornamewithlimits{argmax}}

\newtheorem{theorem}{Theorem}

\theoremstyle{definition}

\theoremstyle{remark}

\title{\vspace{-.8in}On the relationship between predictive coding and backpropagation}
\author{Robert Rosenbaum\footnote{This work was supported by the Air Force Office of Scientific Research (AFOSR) award number FA9550-21-1-0223 and NSF awards DMS-1517828  and DMS-1654268.}
\\{\small \it Department of Applied and Computational} \\{\small \it  Mathematics and Statistics}\\{\small \it University of Notre Dame}\\{\small \it Notre Dame, IN 46556 USA}}
\date{} 

\begin{document}
\maketitle

\noindent {\bf NOTE: }The original published version of this manuscript in PLoS One contains several uncorrected errors in mathematical formulae. These errors all involve instances in which the prediction errors, $\epsilon$, are replaced by the negative prediction errors, $-\epsilon$. A correction to these errors was submitted to PLoS One, but it remains unpublished as of April 16, 2024. \\
\emph{This arxiv version of the manuscript contains the corrected mathematical formulae.}

\section*{Abstract}
Artificial neural networks are often interpreted as abstract models of biological neuronal networks, but they are typically trained using the biologically unrealistic backpropagation algorithm and its variants. Predictive coding has been proposed as a potentially more biologically realistic alternative to backpropagation for training neural networks. This manuscript reviews and extends recent work on the mathematical relationship between predictive coding and backpropagation for training feedforward artificial neural networks on supervised learning tasks. Implications of these results for the interpretation of predictive coding and deep neural networks as models of biological learning are discussed along with a repository of functions, Torch2PC, for performing predictive coding with PyTorch neural network models.


\section*{Introduction}

The backpropagation algorithm and its variants are widely used to train artificial neural networks. While artificial and biological neural networks share some common features, a direct implementation of backpropagation in the brain is often considered biologically implausible in part because of the nonlocal nature of parameter updates: The update to a parameter in one layer depends on activity in all deeper layers. In contrast, biological neural networks are believed to learn largely through local synaptic plasticity rules for which changes to a synaptic weight depend on neural activity local to that synapse. While neuromodulators can have non-local impact on synaptic plasticity, they are not believed to be sufficiently specific to implement the precise, high-dimensional credit assignment required by backpropogation.  However, some work has shown that global errors and neuromodulators can work with local plasticity to implement effective learning algorithms~\cite{izhikevich2007solving,clark2021credit}. 
Backpropagation can be performed using local updates if gradients of neurons' activations are passed upstream through feedback connections, but this interpretation implies other biologically implausible properties of the network, like symmetric feedforward and feedback weights. 
See previous work~\cite{lillicrap2020backpropagation,whittington2019theories} for a more complete review of the biological plausibility of backpropagation.

Several approaches have been proposed for achieving or approximating backpropagation with ostensibly more biologically realistic learning rules~\cite{urbanczik2014learning,lillicrap2016random,scellier2017equilibrium,lillicrap2020backpropagation,whittington2019theories,aljadeff2019cortical,kunin2020two,payeur2021burst,clark2021credit, whittington2017approximation,millidge2020predictive,song2020can,salvatori2021predictive}. One such approach~\cite{ whittington2017approximation,millidge2020predictive,song2020can,salvatori2021predictive} is derived from the theory of ``predictive coding'' or ``predictive processing''~\cite{rao1999predictive,friston2010free,huang2011predictive,bastos2012canonical,clark2015surfing,buckley2017free,bogacz2017tutorial,spratling2017review,keller2018predictive}. 
A relationship between predictive coding and backpropagation was first discovered by Whittington and Bogacz~\cite{whittington2017approximation} who showed that, when predictive coding is used to train a feedforward neural network on a supervised learning task, it can produce parameter updates that approximate those computed by backpropagation. These original results have since been extended to more general network architectures and to show that modifying predictive coding by a ``fixed prediction assumption'' leads to an algorithm that produces the exact same parameter updates as backpropagation~\cite{millidge2020predictive,salvatori2021predictive,song2020can}. 

This manuscript reviews and extends previous work~\cite{whittington2017approximation,millidge2020predictive,song2020can,salvatori2021predictive} on the relationship between predictive coding and backpropagation, as well as some implications of these results on the interpretation of predictive coding and artificial neural networks as models of biological learning. 
The main results in this manuscript are as follows,
\begin{enumerate}
\item Accounting for covariance or precision matrices in hidden layers does not affect parameter updates (learning) for predictive coding under the ``fixed prediction assumption'' used in previous work. 
\item Predictive coding under the fixed prediction assumption is {\it algorithmically} equivalent to a direct implementation of backpropagation, which raises the question of whether it should be interpreted as more biologically plausible than backpropagation. 
\item Empirical results show that the magnitude of prediction errors do not necessarily correspond to surprising features of inputs. 
\end{enumerate}
In addition, a public repository of Python functions, \texttt{Torch2PC}, is introduced. These functions can be used to perform predictive coding on any PyTorch Sequential model.
 \texttt{Torch2PC} can be found at\\
\noindent \url{https://github.com/RobertRosenbaum/Torch2PC}\\
\noindent More information about the software can be found in Materials and Methods.
\noindent A Google Drive folder with Colab notebooks that produce all figures in this text can be found at\\ 
\noindent 
\url{https://drive.google.com/drive/folders/1m_y0G_sTF-pV9pd2_sysWt1nvRvHYzX0}\\ 
\noindent A copy of the same code is also stored at\\ 
\noindent \url{https://github.com/RobertRosenbaum/PredictiveCodingVsBackProp}

\section*{Results}

\subsection*{A review of the relationship between backpropagation and predictive coding from previous work}

For completeness, let us first review the backpropagation algorithm. Consider a feedforward deep neural network (DNN) defined by 
\begin{equation}\label{E:fwdpass}
\begin{aligned}
\hat v_0&=x\\
\hat v_{\ell}&=f_\ell(\hat v_{\ell-1};\theta_\ell),\;\; \ell=1,\ldots,L
\end{aligned}
\end{equation}
where each $\hat v_\ell$ is a vector or tensor of activations, each $\theta_\ell$ is a set of parameters for layer $\ell$, and $L$ is the network's depth. 
In supervised learning, one seeks to minimize a loss function ${\mathcal L}(\hat y,y)$ where $y$ is a label associated with input, $x$, and 
\[
\hat y=f(x;\theta)=\hat v_L
\]
is the network's output, which depends on parameters $\theta=\{\theta_\ell\}_{\ell=1}^L$. 
The loss is typically minimized using gradient-based optimization methods with gradients computed using automatic differentiation tools based on the backpropagation algorithm. For completeness,  backpropagation  is reviewed in the pseudocode below.

\begin{algorithm}[H]
\caption{A standard implementation of backpropagation. }\label{A:backprop}
 \vspace{.04in}
\begin{algorithmic}
\State  {\bf Given:} Input ($x$) and label ($y$)
\vspace{.04in} 
\State {\color{gray}\# forward pass}
\State $\hat v_0=x$
\For{$\ell=1,\ldots,L$}
            \State $\hat v_\ell=f_\ell(\hat v_{\ell-1};\theta_\ell)$
\EndFor

\vspace{.04in} 
\State {\color{gray}\# backward pass}
\State $\delta_L=\tfrac{\partial {\mathcal L}(\hat v_L,y)}{\partial  \hat v_L}$

\For{$\ell=L-1,\ldots,1$}
		\State $\delta_\ell=\delta_{\ell+1}\tfrac{\partial f_{\ell+1}(\hat v_{\ell};\theta_{\ell+1})}{\partial \hat v_\ell}$		
		\State $d\theta_\ell=-\delta_\ell\tfrac{\partial f_\ell(\hat v_{\ell-1};\theta_\ell)}{\partial \theta_\ell}$\vspace{.05in}
\EndFor
\end{algorithmic}
\end{algorithm}


A direct application of the chain rule and mathematical induction shows that backpropagation computes the gradients,
\[
\delta_\ell=\frac{\partial {\mathcal L}(\hat y,y)}{\partial \hat v_\ell}\;\textrm{ and }\;
d\theta_\ell=-\frac{\partial {\mathcal L}(\hat y,y)}{\partial  \theta_\ell}.
\]
The negative gradients, $d\theta_\ell$, are then used to update parameters, either directly for stochastic gradient descent or indirectly for other gradient-based learning methods~\cite{goodfellow2016deep}. For the sake of comparison, {I} used backpropagation to train a 5-layer convolutional neural network on the MNIST data set (Figure~\ref{F:PC}A,B; blue curves). 
{I} next review algorithms derived from the theory of predictive coding and their relationship to backpropagation, as originally derived in previous work~\cite{whittington2017approximation,millidge2020predictive,salvatori2021predictive,song2020can}.

\begin{figure}[!h]
\centering{
\includegraphics[width=4.5in]{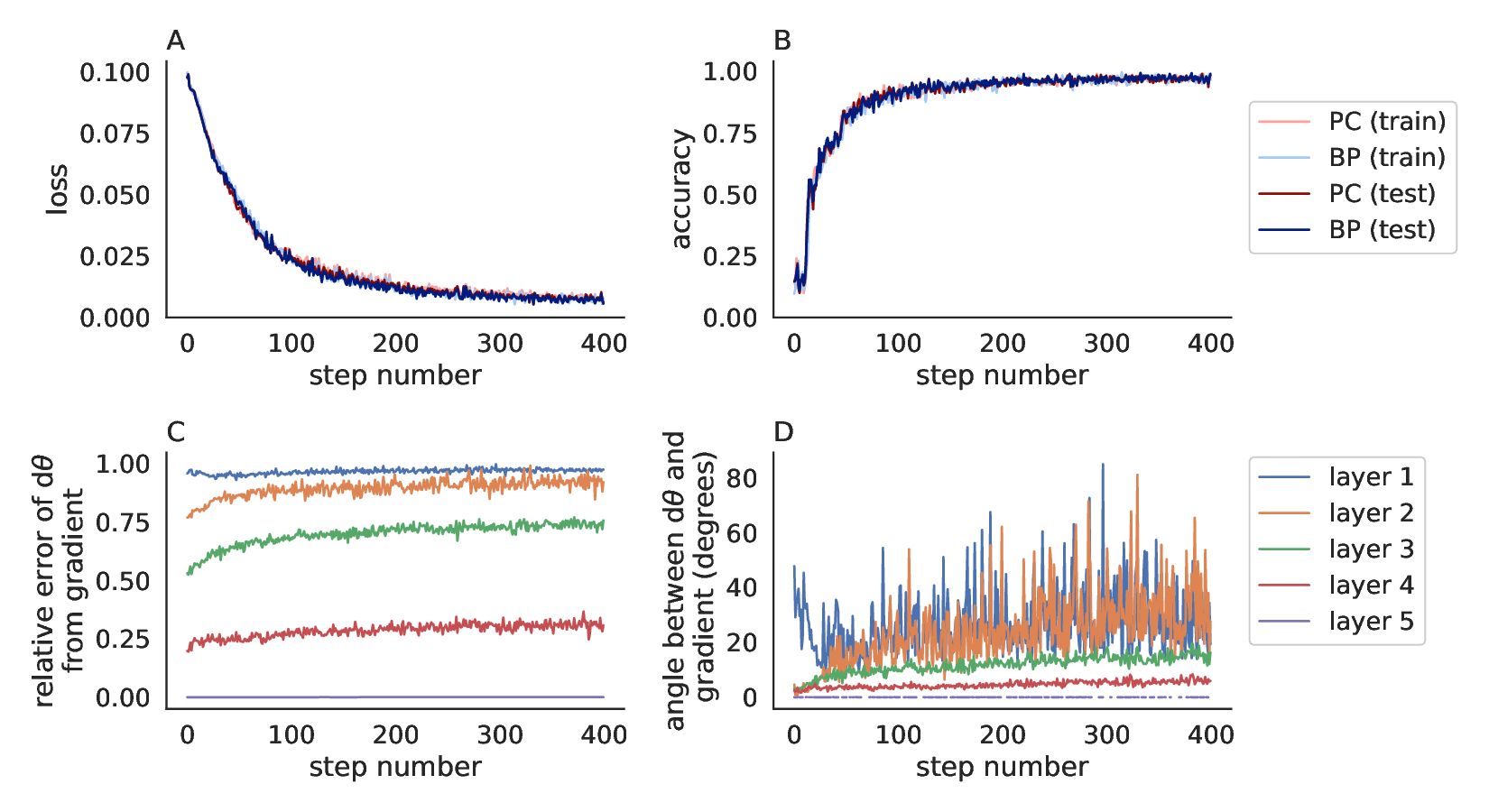}
 }
 \caption{{\bf Comparing backpropagation and predictive coding in a convolutional neural network trained on MNIST.} 
{\bf A,B)} The loss (A) and accuracy (B) on the training set (pastel) and test set (dark) when a 5-layer network was trained using a strict implementation of predictive coding (Algorithm~\ref{A:PC} with $\eta=0.1$ and $n=20$; red) and backpropagation (blue). {\bf C,D)} The relative error (C) and angle (B) between the parameter update, $d\theta$, computed by Algorithm~\ref{A:PC} and the negative gradient of the loss at each layer. Predictive coding and backpropagation give similar accuracies, but the parameter updates are less similar.
}
 \label{F:PC}
 \end{figure}

\subsubsection*{A strict interpretation of predictive coding does not accurately compute gradients.}

{I} begin by reviewing supervised learning under a strict interpretation of predictive coding. The formulation in this section is equivalent to the one first studied by Whittington and Bogacz~\cite{whittington2017approximation}, except that their results are restricted to the case in which $f_\ell(v_{\ell-1};\theta_\ell)=\theta_\ell g_\ell(v_{\ell-1})$ for some point-wise-applied activation function, $g_\ell$, and connectivity matrix, $\theta_\ell$. Our formulation  extends this formulation to arbitrary vector-valued differentiable functions, $f_\ell$. For the sake of continuity with later sections, {I} also use the notational conventions from~\cite{millidge2020predictive} which differ from those in~\cite{whittington2017approximation}. 

Predictive coding can be derived from a hierarchical, Gaussian probabilistic model in which each layer, $\ell$, is associated with  a Gaussian random variable, $V_\ell$, satisfying
\begin{equation}\label{E:ProbModel}
\begin{aligned}
p(V_\ell=v_\ell\,|\,V_{\ell-1}= v_{\ell-1})=\mathcal N(v_\ell\, ; \,\,f_\ell( v_{\ell-1};\theta_\ell),\;\Sigma_\ell)
\end{aligned}
\end{equation}
where $\mathcal N(v;\,\mu,\Sigma)\propto \exp(-[v-\mu]^T\Sigma^{-1}[v-\mu]/2)$ is the multivariate Gaussian distribution with mean, $\mu$, and covariance matrix, $\Sigma$, evaluated at $v$. Following previous work~\cite{whittington2017approximation,millidge2020predictive,salvatori2021predictive,song2020can}, {I} take $\Sigma=I$ to be the identity matrix, but later relax this assumption~\cite{bogacz2017tutorial}. 

If we condition on an observed input, $V_0=x$, then a forward pass through the network described by Eq.~\eqref{E:fwdpass} corresponds to setting $\hat v_0=x$ and then sequentially computing the conditional expectations or, equivalently, maximizing conditional probabilities, 
\[
\begin{aligned}
\hat v_\ell&=E[V_\ell\,|\,V_{\ell-1}=\hat v_{\ell-1}]\\
&=\argmax_{v_\ell} p(V_\ell=v_\ell\,|\,V_{\ell-1}=\hat v_{\ell-1})\\
&=f_\ell(\hat v_{\ell-1};\theta_\ell)
\end{aligned}
\]
until reaching an inferred output, $\hat y=\hat v_L$. 
Note that this forward pass does not necessarily maximize the global conditional probability, $p(V_L=\hat y\,|\,v_0=x)$ and it does not account for a prior distribution on $V_L$, which arises in related work on predictive coding for unsupervised learning~\cite{rao1999predictive,bogacz2017tutorial}.
One interpretation of a forward pass is that each $\hat v_\ell$ is the network's ``belief'' about the state of $V_\ell$, when only $V_0=x$ has been observed. 




Now suppose that we condition on both an observed input, $V_0=x$, {\it and} its label, $V_L=y$. In this case, generating beliefs about the hidden states, $V_\ell$, is more difficult because we need to account for potentially conflicting information at each end of the network. We can proceed by initializing a set of beliefs, $v_\ell$, about the state of each $V_\ell$, and then updating our initial beliefs to be more consistent with the observations, $x$ and $y$, and  parameters, $\theta_\ell$. 

The error made by a set of beliefs, $\{v_\ell\}_{\ell=1}^L$, under parameters, $\{\theta_\ell\}_{\ell=1}^L$, can be quantified by 
\[
\epsilon_\ell=f_\ell(v_{\ell-1};\theta_\ell)-v_\ell
\]
for $\ell=1,\ldots,L-1$ where $v_0=V_0=x$ is observed. It is not so simple to quantify the error, $\epsilon_L$, made at the last layer in a way that accounts for arbitrary loss functions. In the special case of a squared-Euclidean loss function,
\[
{\mathcal L}(\hat y,y)=\frac{1}{2}\|\hat y-y\|^2,
\]
where $\|u\|^2=u^Tu$. 
Standard formulations of predictive coding~\cite{bogacz2017tutorial,buckley2017free} use
\begin{equation}\label{E:epsLPC}
\epsilon_L=f_L(v_{L-1};\theta_L)-y
\end{equation}
where recall that $y$ is the label. In this case, $\epsilon_L$ satisfies
\begin{equation}\label{E:epsLPC2}
\epsilon_L=\frac{\partial {\mathcal L}(\tilde v_L,y)}{\partial \tilde v_L}
\end{equation}
where 
\[
\tilde v_L=f_L(v_{L-1};\theta_L).
\] 
We use the $\tilde \cdot$ to emphasize that $\tilde v_L$ is different from $\hat v_L$ (which is defined by a forward pass starting at $\hat v_0=x$) and is defined in a fundamentally different way from the $v_\ell$ terms (which do not necessarily satisfy $v_\ell=f_\ell(v_{\ell-1};\theta_\ell)$). 
We can then define the total summed magnitude of errors as
\[
F=\frac{1}{2}\sum_{\ell=1}^L \|\epsilon_\ell\|^2
\]
More details on the derivation of $F$ in terms of variational Bayesian inference can be found in previous work~\cite{friston2010free,bogacz2017tutorial,buckley2017free,millidge2020predictive} where $F$ is known as the variational free energy of the model. Essentially, minimizing $F$ produces a model that is more consistent with the observed data. Minimizing $F$ by gradient descent on $v_\ell$ and $\theta_\ell$ produce the inference and learning steps of predictive coding, respectively. 

Under a more heuristic interpretation,  $v_\ell$ represents the network's ``belief'' about $V_\ell$, and $f_\ell(v_{\ell-1};\theta_\ell)$ is the ``prediction'' of $v_\ell$ made by the previous layer. Under this interpretation, $\epsilon_\ell$ is the error made by the previous layer's prediction, so $\epsilon_\ell$ is called a ``prediction error.'' Then $F$ quantifies the total magnitude of  prediction errors given a set of beliefs, $v_\ell$, parameters, $\theta_\ell$, and observations, $V_0=x$ and $V_L=y$.


In predictive coding, beliefs, $v_\ell$, are updated to minimize the error, $F$. This can be achieved by gradient descent, \emph{i.e.}, by making updates of the form 
\[
v_\ell \gets v_\ell +\eta dv_\ell
\]
where $\eta$ is a step size and 
\begin{equation}\label{E:dvell}
\begin{aligned}
dv_\ell&=-\frac{\partial F}{\partial v_\ell}\\
&=\epsilon_\ell-\epsilon_{\ell+1}\frac{\partial f_{\ell+1}(v_{\ell};\theta_{\ell+1})}{\partial v_\ell}
\end{aligned}
\end{equation}
In this expression, ${\partial f_{\ell+1}(v_{\ell};\theta_{\ell+1})}/{\partial v_\ell}$ is a Jacobian matrix and $\epsilon_{\ell+1}$  is a row-vector to simplify notation, but a column-vector interpretation is similar. If $x$ is a mini-batch instead of one data point, then $v_\ell$ is an $m\times n_\ell$ matrix and derivatives are tensors.  These conventions are used throughout the manuscript. 
The updates in Eq.~\eqref{E:dvell} can be iterated until convergence or approximate convergence. Note that the prediction errors, $\epsilon_\ell=f_\ell(v_{\ell-1};\theta_\ell)-v_\ell$, should also be updated on each iteration. 

Learning can also be phrased as minimizing $F$ with gradient descent on parameters. Specifically, 
\[
\theta_\ell=\theta_\ell+\eta_\theta d\theta_\ell
\]
where
\begin{equation}\label{E:dtheta}
\begin{aligned}
d\theta_\ell&=-\frac{\partial F}{\partial \theta_\ell}\\
&=-\epsilon_\ell\frac{\partial f_\ell(v_{\ell-1};\theta_\ell)}{\partial \theta_\ell}.
\end{aligned}
\end{equation}
Note that some previous work uses the negative of the prediction errors used here, {\it i.e.}, they use $\epsilon_\ell=v_\ell-f_\ell(v_{\ell-1};\theta_\ell)$. While this choice changes some of the expressions above, the value of $F$ and its dependence on $\theta_\ell$ is not changed because $F$ is defined by the norms of the $\epsilon_\ell$ terms. The complete algorithm is defined more precisely by the pseudocode below:


\begin{algorithm}[H]
\caption{A direct interpretation of predictive coding.}\label{A:PC}
 \vspace{.04in}
\begin{algorithmic}
\State  {\bf Given:} Input ($x$), label ($y$), and initial beliefs ($v_\ell$)


\vspace{.06in} 
\State {\color{gray}\# error and belief computation}
\For{$i=1,\ldots,n$}
	\State $\tilde v_L=f_L(v_{L-1};\theta_L)$
	\State $\epsilon_L=\tfrac{\partial {\mathcal L}(\tilde v_L,y)}{\partial \tilde v_L}$
	\For{$\ell=L-1,\ldots,1$}
		\State $\epsilon_\ell=f_\ell(v_{\ell-1};\theta_\ell)-v_\ell$		
	
		\State $dv_\ell=\epsilon_\ell-\epsilon_{\ell+1}\tfrac{\partial f_{\ell+1}(v_{\ell};\theta_{\ell+1})}{\partial v_\ell}$
		\State $v_\ell=v_\ell+\eta dv_\ell$ 
	\EndFor
\EndFor

\vspace{.06in} 
\State {\color{gray}\# parameter update computation}
\For{$\ell=1,\ldots,L$}
\State $d\theta_\ell=-\epsilon_\ell\tfrac{\partial f_\ell(v_{\ell-1};\theta_\ell)}{\partial \theta_\ell}$\vspace{.05in}
\EndFor
\end{algorithmic}
\end{algorithm}

Here and elsewhere, $n$ denotes the number of iterations for the inference step.
The choice of initial beliefs is not specified in the algorithm above, but previous work~\cite{whittington2017approximation,millidge2020predictive,salvatori2021predictive,song2020can} uses the results from a forward pass, $v_\ell=\hat v_\ell$, as initial conditions and {I} do the same in all numerical examples. 

{I} tested Algorithm~\ref{A:PC} on MNIST using a 5-layer convolutional neural network. To be consistent with the definitions above, {I} used a mean-squared error (squared Euclidean) loss function, which required one-hot encoded labels~\cite{goodfellow2016deep}. Algorithm~\ref{A:PC} performed similarly to backpropagation  (Fig.~\ref{F:PC}A,B) even though the parameter updates  did not match the true gradients  (Fig.~\ref{F:PC}C,D). Algorithm~\ref{A:PC} was slower than backpropagation (31s for Algorithm~\ref{A:PC} versus 8s for backpropagation when training metrics were not  computed on every iteration) in part because Algorithm~\ref{A:PC} requires several inner iterations to compute the prediction errors ($n=20$ iterations used in this example).  Algorithm~\ref{A:PC} failed to converge on a larger model. Specifically, the loss grew consistently with iterations when trying to use Algorithm~\ref{A:PC} to train the 6-layer CIFAR-10 model described in the next section.
\nameref{S:1} shows the same results from Fig.~\ref{F:PC} repeated across 30 trials with different random seeds to quantify the mean and standard deviation across trials.

\begin{figure}[!h]
 \includegraphics[width=5in]{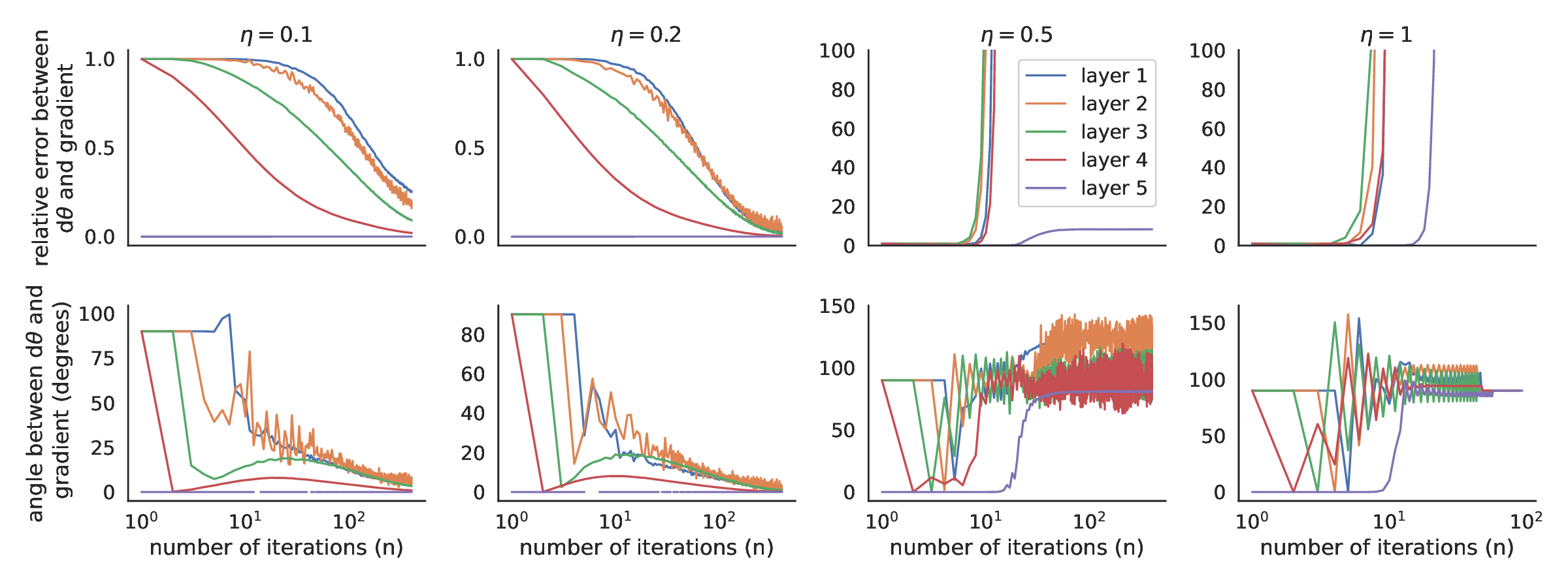}
\caption{{\bf Comparing parameter updates from predictive coding to true gradients in a network trained on MNIST.}  Relative error and angle between $d\theta_\ell$ produced by predictive coding (Algorithm~\ref{A:PC}) as compared to the exact gradients, $\partial {\mathcal L}/\partial \theta_\ell$ computed by backpropagation (relative error defined by $\|d\theta_{pc}-d\theta_{bp}\|/\|d\theta_{bp}\|$). Updates were computed as a function of the number of iterations, $n$, used in Algorithm~\ref{A:PC} for various values of the step size, $\eta$, using the model from Fig.~\ref{F:PC} applied to one mini-batch of data. Both models were initialized identically to the pre-trained parameter values from the trained model in Fig.~\ref{F:PC}. Parameter updates converge near the gradients after many iterations for smaller values of $\eta$, but diverge for larger values. 
}
\label{F:PCErrs}
\end{figure}

Fig.~\ref{F:PC}C,D shows that predictive coding does not update parameters according to the true gradients, but it is not immediately clear whether this would be resolved by using more iterations (larger $n$) or different values of the step size, $\eta$. {I} next compared the parameter updates, $d\theta_\ell$, to the true gradients, $\partial {\mathcal L}/\partial \theta_\ell$ for different values of $n$ and $\eta$ (Fig.~\ref{F:PCErrs}). For the smaller values of $\eta$ tested ($\eta=0.1$ and $\eta=0.2$) and larger values of $n$ ($n>100$), parameter updates were similar to the true gradients in the last two layers, but they differed substantially in the first two layers. The largest  values of $\eta$ tested ($\eta=0.5$ and $\eta=1$) caused the iterations in Algorithm~\ref{A:PC} to diverge. 


Some choices in designing Algorithm~\ref{A:PC} were made arbitrarily. For example, the three updates inside the inner for-loop over $\ell$ could be performed in a different order or the outer for-loop over $i$  could be changed to a while-loop with a convergence criterion.   For any initial conditions and any of these design choices, if the iterations over $i$ are repeated until convergence or approximate convergence of each $v_\ell$ to a fixed point, $v_\ell^*$, then the increments must satisfy $dv_\ell=0$ at the fixed point and therefore the fixed point values of the prediction errors, $\epsilon^*_\ell$, must satisfy
\begin{equation}\label{E:EpsStarPC}
\epsilon^*_\ell=\frac{\partial f_{\ell+1}(v^*_{\ell};\theta_{\ell+1})}{\partial v^*_\ell}\epsilon^*_{\ell+1}
\end{equation}
for $\ell=1,\ldots,L-1$. By the definition of $\epsilon_L$, we have
\begin{equation}\label{E:EpsLStarPC}
\epsilon_L^*=\frac{\partial {\mathcal L}(\tilde v^*_L,y)}{\partial  \tilde v^*_L}.
\end{equation}
where
\[
\tilde v^*_L=f_L(v^*_{L-1};\theta_L).
\]
Combining Eqs.~\eqref{E:EpsStarPC} and \eqref{E:EpsLStarPC} gives the fixed point prediction errors of the penultimate layer
\begin{equation}\label{E:123}
\begin{aligned}
\epsilon_{L-1}^*&=\frac{\partial {\mathcal L}(\tilde v^*_L,y)}{\partial  \tilde v^*_L}\frac{\partial f_{L}(v^*_{L-1};\theta_{L})}{\partial v^*_{L-1}}\\
&=\frac{\partial {\mathcal L}(\tilde v^*_L,y)}{\partial v^*_{L-1}}
\end{aligned}
\end{equation}
where we used the fact that $  \tilde v^*_L=f_{L}(v^*_{L-1};\theta_{L})$ and the chain rule. The error in layer $L-2$ is given by 
\[
\begin{aligned}
\epsilon_{L-2}^*&=\frac{\partial {\mathcal L}(\tilde v^*_L,y)}{\partial v^*_{L-1}}\frac{\partial f_{L-1}(v^*_{L-2};\theta_{L-1})}{\partial v^*_{L-2}}.
\end{aligned}
\]
Note that we cannot apply the chain rule to reduce this product (like we did for Eq.~\eqref{E:123}) because it is not necessarily true that $ v^*_{L-1}= f_{L-1}(v^*_{L-2};\theta_{L-1})$. {I} revisit this point below. 
We can continue this process to derive 
\[
\epsilon_{L-3}^*=\frac{\partial {\mathcal L}(\tilde v^*_L,y)}{\partial v^*_{L-1}}\frac{\partial f_{L-1}(v^*_{L-2};\theta_{L-1})}{\partial v^*_{L-2}}
\frac{\partial f_{L-2}(v^*_{L-3};\theta_{L-2})}{\partial v^*_{L-3}}
\]
and continue for $\ell=L-4,\ldots,1$. 
In doing so, we see (by induction) that  $\epsilon_\ell^*$ can be written as 
\begin{equation}\label{E:epsilonellStar}
\begin{aligned}
\epsilon_{\ell}^*&=\frac{\partial {\mathcal L}(\tilde v^*_L,y)}{\partial  \tilde v^*_{L-1}}\prod_{\ell'=\ell}^{L-2} \frac{\partial f_{\ell'+1}(v^*_{\ell'};\theta_{\ell'+1})}{\partial v^*_{\ell'}}.
\end{aligned}
\end{equation}
for $\ell=1,\ldots,L-2$. 
Therefore, if the inference loop converges to a fixed point, then the subsequent parameter update obeys
\begin{equation}\label{E:dthetaStar}
d\theta_\ell=-\frac{\partial {\mathcal L}(\tilde v^*_L,y)}{\partial  \tilde v^*_{L-1}}\left[\prod_{\ell'=\ell}^{L-2} \frac{\partial f_{\ell'+1}(v^*_{\ell'};\theta_{\ell'+1})}{\partial v^*_{\ell'}}\right]\frac{\partial f_\ell(v^*_{\ell-1};\theta_\ell)}{\partial \theta_\ell}
\end{equation}
by Eq.~\eqref{E:dtheta}. 
It is not clear whether  there is a simple mathematical relationship between these parameter updates and the negative gradients, $d\theta_\ell=-\partial {\mathcal L}/\partial \theta_\ell$, computed by backpropagation.

It is tempting to assume that $v^*_{\ell}=f_{\ell}(v^*_{\ell-1};\theta_{\ell})$, in which case the product terms would be reduced by the chain rule. Indeed, this assumption would imply that $v^*_\ell=\hat v_\ell$ and $\tilde v_L^*=\hat v_L$ and, finally, that $\epsilon_\ell=\partial {\mathcal L}/\partial \hat v_\ell$  and $d\theta_\ell=-\partial {\mathcal L}/\partial \theta_\ell$, identical to the values computed by backpropagation. However,  we cannot generally expect to have $v^*_{\ell}=f_{\ell}(v^*_{\ell-1};\theta_{\ell})$ because this would imply that $\epsilon_{\ell}^*=0$ and therefore $\partial {\mathcal L}/\partial v_\ell^*=\partial {\mathcal L}/\partial \theta_\ell=0$. In other words, Algorithm~\ref{A:PC} is only equivalent to backpropagation in the case where parameters are at a critical point of the loss function, so all updates are zero. 
Nevertheless, this thought experiment suggests a modification to Algorithm~\ref{A:PC}  for which the fixed points {\it do} represent the true gradients~\cite{millidge2020predictive,whittington2017approximation}. {I} review that modification in the next section.

Note also that the calculations above rely on the assumption of a Euclidean loss function, ${\mathcal L}(\hat y,y)=\|\hat y-y\|^2/2$. If we want to generalize the algorithm to different loss functions, then Eqs.~\eqref{E:epsLPC} and \eqref{E:epsLPC2} could not both be true, and therefore Eqs.~\eqref{E:EpsStarPC} and \eqref{E:EpsLStarPC} could not both be true. This leaves open the question of how to define $\epsilon_L$ when using loss functions that are not proportional to the squared Euclidean norm. 
If we were to define $\epsilon_L$ by \eqref{E:epsLPC}, at the expense of losing \eqref{E:epsLPC2}, then the algorithm would not account for the loss function at all, so it would effectively assume a Euclidean loss, {\it i.e.}, it would compute the same values that are computed by Algorithm~\ref{A:PC} with a Euclidean loss. 
If we instead were to define $\epsilon_L$ by Eq.~\eqref{E:epsLPC2} at the expense of \eqref{E:epsLPC}, then Eqs.~\eqref{E:dvell} and \eqref{E:EpsStarPC} would no longer be true for $\ell=L-1$ and Eq.~\eqref{E:dtheta} would no longer be true for $\ell=L$. Instead, all three of these equations would involve second-order derivatives of the loss function, and therefore the fixed point equations \eqref{E:epsilonellStar} and \eqref{E:dthetaStar} would also involve second order derivatives. The interpretation of the parameter updates is not clear in this case.
One might instead try to define $\epsilon_L$ by the result of a forward pass,
\[
\begin{aligned}
\epsilon_L&=f_L(\hat v_{L-1};\theta_L)-y\\
&=\hat v_L-y
\end{aligned}
\]
but then $\epsilon_L$ would be a constant with respect to $v_{L-1}$, so we would have $\partial \epsilon_L/\partial v_{L-1}=0$, and therefore  Eq.~\eqref{E:dvell} at $\ell=L-1$ would become
\[
\begin{aligned}
dv_{L-1}&=-\frac{\partial F}{\partial v_{L-1}}\\
&=\epsilon_{L-1}
\end{aligned}
\]
which has a fixed point at $\epsilon^*_{L-1}=0$. This would finally imply that all the errors converge to $\epsilon_\ell^*=0$ and therefore $d\theta_\ell=0$ at the fixed point.

{I} next discuss a modification of Algorithm~\ref{A:PC} that converges to the same gradients computed by backpropagation, {\it and} is applicable to general loss functions~\cite{millidge2020predictive,whittington2017approximation}.

\subsubsection*{Predictive coding modified by the fixed prediction assumption converges to the gradients computed by backpropagation.}

 \begin{figure}
\centering{
 \includegraphics[width=4.5in]{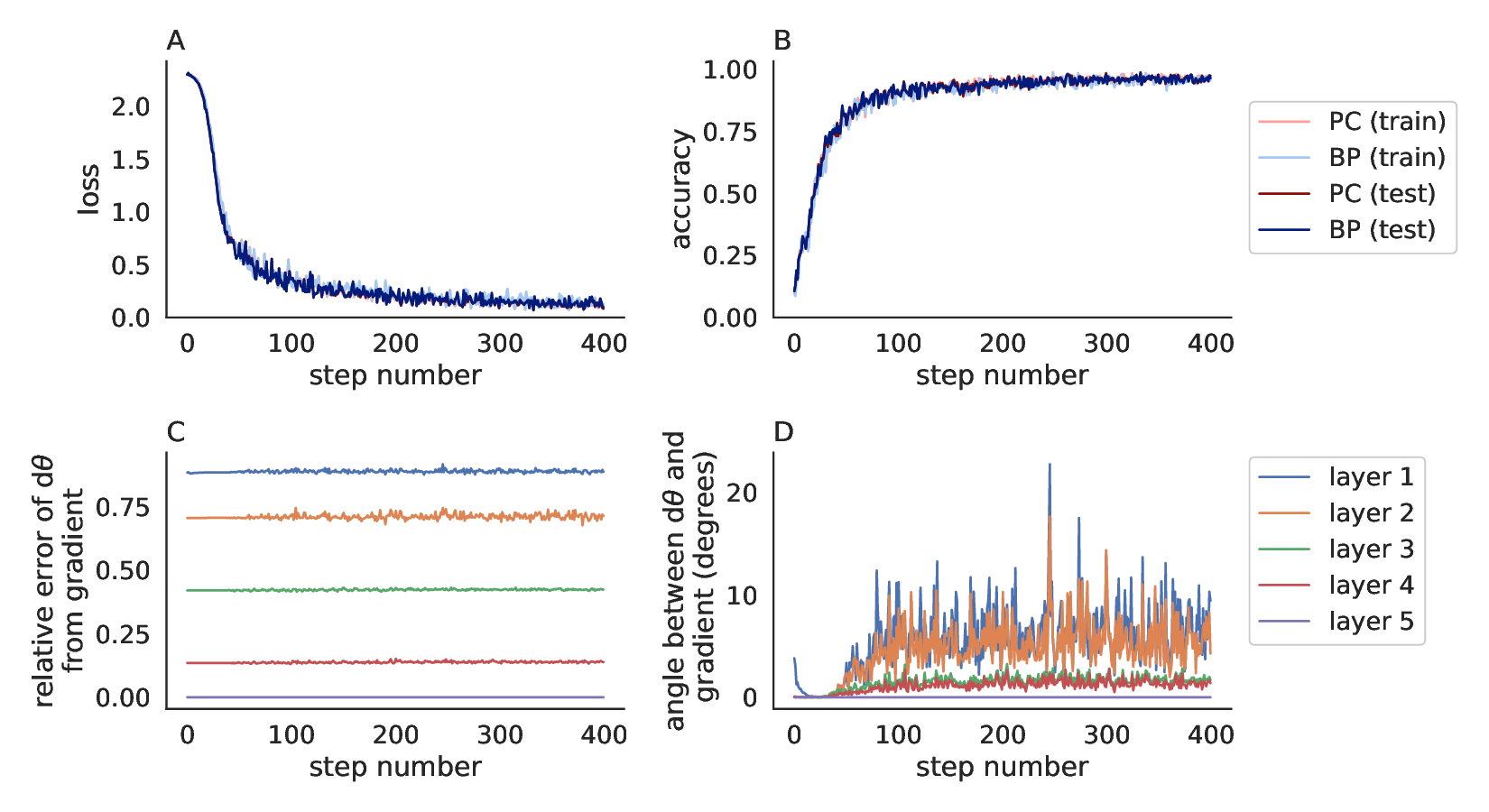}
}
 \caption{{\bf Predictive coding modified by the fixed prediction assumption compared to backpropagation in a convolutional neural network trained on MNIST.} Same as Fig.~\ref{F:PC} except Algorithm~\ref{A:ModPC} was used (with $\eta=0.1$ and $n=20$) in place of Algorithm~\ref{A:PC}.  The accuracy of predictive coding with the fixed prediction assumption is similar to backpropagation, but the parameter updates are less similar for these hyperparameters. }
 \label{F:ModPC}
 \end{figure}

Previous work~\cite{whittington2017approximation,millidge2020predictive} proposed a modification of the predictive coding algorithm described above called the ``fixed prediction assumption'' which {I} now review. Motivated by the considerations in the last few paragraphs of the previous section, we can selectively substitute some terms of the form $v_\ell$ and $f_\ell (v_{\ell-1};\theta_\ell)$ in Algorithm~\ref{A:PC} with $\hat v_\ell$ (or, equivalently, $f_\ell (\hat v_{\ell-1};\theta_\ell)$) where $\hat v_\ell$ are the results of the original forward pass starting from $\hat v_0=x$. Specifically,  the following modifications are made to the quantities computed by Algorithm~\ref{A:PC}
\begin{equation}\label{E:ModPC}
\begin{aligned}
\epsilon_\ell&=\hat v_\ell-v_\ell\\
\epsilon_L&=\frac{\partial {\mathcal L}(\hat v_L,y)}{\partial \hat v_L}\\
dv_\ell&=\epsilon_\ell-\epsilon_{\ell+1}\frac{\partial f_{\ell+1}(\hat v_{\ell};\theta_{\ell+1})}{\partial \hat v_\ell}\\
d\theta_\ell&=-\epsilon_\ell\frac{\partial f_\ell(\hat v_{\ell-1};\theta_\ell)}{\partial \theta_\ell}
\end{aligned}
\end{equation}
for $\ell=1,\ldots,L-1$. This modification can be interpreted as ``fixing'' the predictions at the values computed by a forward pass and is therefore called the ``fixed prediction assumption''~\cite{whittington2017approximation,millidge2020predictive}. 
Additionally, the initial conditions of the beliefs are set to the results from a forward pass,
$
v_\ell=\hat v_\ell
$
for $\ell=1,\ldots,L-1$. 
The complete modified algorithm is defined by the pseudocode below:



\begin{algorithm}[H]
\caption{Supervised learning with predictive coding modified by the fixed prediction assumption. Adapted from the algorithm in~\cite{millidge2020predictive} and similar to the algorithm from~\cite{whittington2017approximation}.}\label{A:ModPC}
\begin{algorithmic}
 \vspace{.04in}
\State {\bf Given:} Input ($x$) and label ($y$)
\vspace{.04in} 
\State {\color{gray}\# forward pass}
\State $\hat v_0=x$
\For{$\ell=1,\ldots,L$}
            \State $\hat v_\ell=f_\ell(\hat v_{\ell-1};\theta_\ell)$
            \State $v_\ell=\hat v_\ell$
\EndFor

\vspace{.04in} 
\State {\color{gray}\# error and belief computation}
\State $\epsilon_L=\tfrac{\partial {\mathcal L}(\hat v_L,y)}{\partial  \hat v_L}$
\For{$i=1,\ldots,n$}
	\For{$\ell=L-1,\ldots,1$}
		\State $\epsilon_\ell=\hat v_\ell- v_\ell$
		\State $dv_\ell=\epsilon_\ell-\epsilon_{\ell+1}\tfrac{\partial f_{\ell+1}(\hat v_{\ell};\theta_{\ell+1})}{\partial \hat v_\ell}$
		\State $v_\ell=v_\ell+\eta dv_\ell$ 
	\EndFor
\EndFor

\vspace{.06in} 
\State {\color{gray}\# parameter update computation}
\For{$\ell=1,\ldots,L$}
\State $d\theta_\ell=-\epsilon_\ell\tfrac{\partial f_\ell(\hat v_{\ell-1};\theta_\ell)}{\partial \theta_\ell}$\vspace{.05in}
\EndFor
\end{algorithmic}
\end{algorithm}

Note, again, that some choices in Algorithm~\ref{A:ModPC} were made arbitrarily. The three updates inside the inner for-loop over $\ell$ could be performed in a different order or the outer for loop over $i$  could be changed to a while-loop with a convergence criterion. 
Regardless of these choices, the fixed points, $\epsilon^*_\ell$, can again be computed by setting $dv_\ell=0$ to obtain
\[
\epsilon^*_\ell=\epsilon^*_{\ell+1}\frac{\partial f_{\ell+1}(\hat v_{\ell};\theta_{\ell+1})}{\partial \hat v_\ell}.
\]
Now note that $\epsilon_L$ is fixed, so 
\[
\epsilon^*_L=\frac{\partial {\mathcal L}(\hat v_L,y)}{\partial  \hat v_L}
\]
and we can combine these two equations to compute 
\[
\begin{aligned}
\epsilon^*_{L-1}&=\frac{\partial {\mathcal L}(\hat v_L,y)}{\partial  \hat v_L}\frac{\partial f_{L}(\hat v_{L-1};\theta_{L})}{\partial \hat v_{L-1}}\\
&=\frac{\partial {\mathcal L}(\hat v_L,y)}{\partial \hat v_{L-1}}
\end{aligned}
\]
where we used the chain rule and the fact that $\hat v_\ell=f_\ell(\hat v_{\ell-1};\theta_\ell)$. Continuing this approach we have,
\[
\begin{aligned}
\epsilon^*_\ell&=\epsilon^*_{\ell+1}\frac{\partial f_{\ell+1}(\hat v_{\ell};\theta_{\ell+1})}{\partial \hat v_\ell}\\
&=\frac{\partial {\mathcal L}(\hat y,y)}{\partial \hat v_{\ell}}
\end{aligned}
\]
for all $\ell=1,\ldots,L$ (where recall that $\hat y=\hat v_L$ is the output from the feedfoward pass). 
Combining this with the modified definition of $d\theta_\ell$, we have
\[
\begin{aligned}
d\theta_\ell&=-\epsilon^*_\ell\frac{\partial f_\ell(\hat v_{\ell-1};\theta_\ell)}{\partial \theta_\ell}\\
&=-\frac{\partial {\mathcal L}(\hat y,y)}{\partial \hat v_{\ell}}\frac{\partial \hat v_\ell}{\partial \theta_\ell}\\
&=-\frac{\partial {\mathcal L}(\hat y,y)}{\partial \theta_\ell}
\end{aligned}
\]
where we use the chain rule and the fact that $\hat v_\ell=f_\ell(\hat v_{\ell-1};\theta_\ell)$. We may conclude that, if the inference step converges to a fixed point ($dv_\ell=0$), then Algorithm~\ref{A:ModPC}  computes the same values of $d\theta_\ell$ as backpropagation and also that the prediction errors, $\epsilon_\ell$, converge to the gradients, $\delta_\ell=\partial {\mathcal L}/\partial \hat v_\ell$, computed by backpropagation. As long as the inference step  {\it approximately} converges to a fixed point ($dv_\ell\approx 0$), then we should expect the parameter updates from Algorithm~\ref{A:ModPC} to {\it approximate} those computed by backpropagation. In the next section, {I} extend this result to show that a special case of the algorithm computes the true gradients in a fixed number of steps.

{I} next tested Algorithm~\ref{A:ModPC} on MNIST using the same 5-layer convolutional neural network considered above. {I} used a cross-entropy loss function, but otherwise used all of the same parameters used to test Algorithm~\ref{A:PC} in Fig.~\ref{F:PC}. The modified predictive coding algorithm (Algorithm~\ref{A:ModPC}) performed similarly to backpropagation in terms of the loss and accuracy (Fig.~\ref{F:ModPC}A,B). Parameter updates computed by Algorithm~\ref{A:ModPC} did not match the true gradients, but pointed in a similar direction and provided a closer match than Algorithm~\ref{A:PC} (compare Fig.~\ref{F:ModPC}C,D to Fig.~\ref{F:PC}C,D).  Algorithm~\ref{A:ModPC} was similar to Algorithm~\ref{A:PC} in terms of training time (29s for  Algorithm~\ref{A:ModPC} versus 31s for Algorithm~\ref{A:PC} and 8s for backpropagation). 
\nameref{S:2} shows the same results from Fig.~\ref{F:ModPC} repeated across 30 trials with different random seeds to quantify the mean and standard deviation across trials.

\begin{figure}
\centering{
  \includegraphics[width=5in]{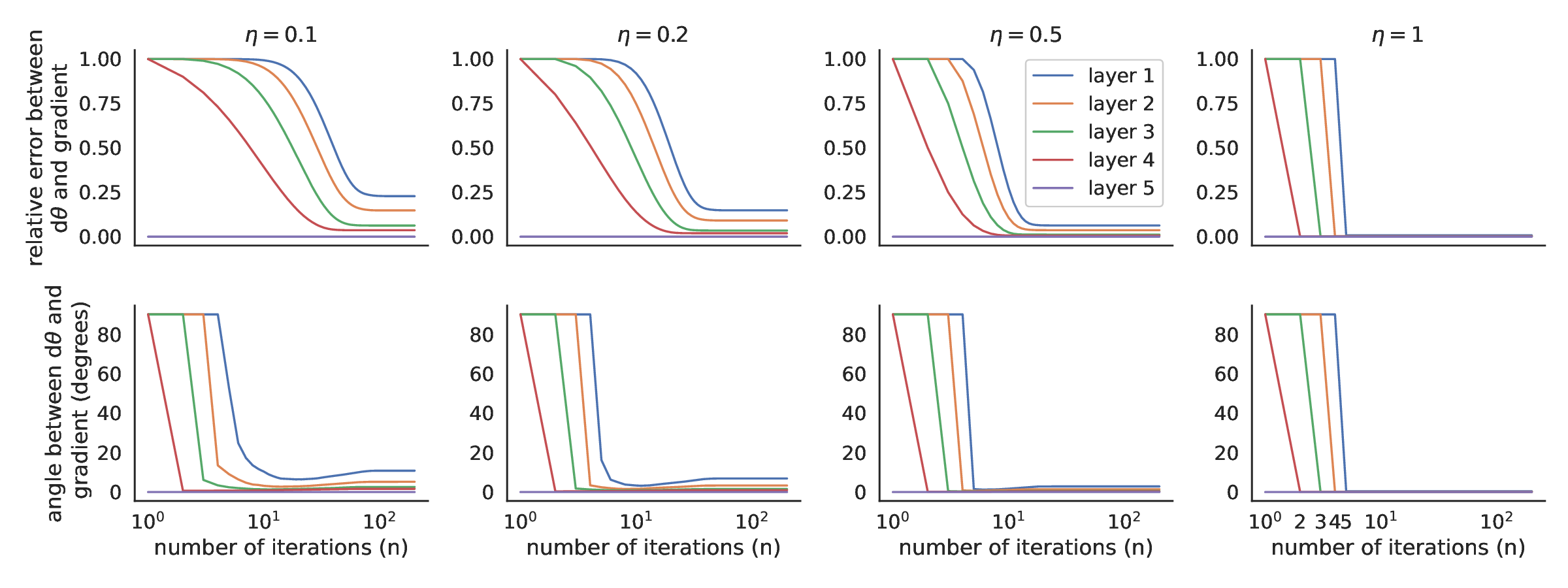}
}
 \caption{{\bf Comparing parameter updates from predictive coding modified by the fixed prediction assumption to true gradients in a network trained on MNIST.} Relative error and angle between $d\theta$ produced by predictive coding modified by the fixed prediction assumption (Algorithm~\ref{A:ModPC}) as compared to the exact gradients computed by backpropagation (relative error defined by $\|d\theta_{pc}-d\theta_{bp}\|/\|d\theta_{bp}\|$). Updates were computed as a function of the number of iterations, $n$, used in Algorithm~\ref{A:ModPC} for various values of the step size, $\eta$, using the model from Fig.~\ref{F:ModPC} applied to one mini-batch of data. Both models were initialized identically to the pre-trained parameter values from the backpropagation-trained model in Fig.~\ref{F:ModPC}. In the rightmost panels, some lines are not visible where they overlap at zero. Parameter updates quickly converge to the true gradients when $\eta$ is larger.
}
 \label{F:ModPCErrs}
 \end{figure}

{I} next compared the parameter updates computed by Algorithm~\ref{A:ModPC} to the true gradients for different values of $n$ and $\eta$ (Fig.~\ref{F:ModPCErrs}).    When  $\eta<1$, the parameter updates, $d\theta_\ell$, appeared to converge, but did not converge exactly to the true gradients. This is likely due to numerical floating point errors accumulated over  iterations. When $\eta=1$, the parameter updates at each layer remained constant for the first few iterations, then immediately jumped to become very near the updates from backpropagation. In the next section, {I} provide a mathematical analysis of this behavior and show that when $\eta=1$, Algorithm~\ref{A:ModPC} computes the true gradients in a fixed number of steps.

To see how well these results extend to a larger model and more difficult benchmark, {I} next tested Algorithm~\ref{A:ModPC} on CIFAR-10~\cite{krizhevsky2009learning} using a six-layer convolutional network. While the network only had one more layer than the MNIST network used above, it had 141 times more parameters (32,695 trainable parameters in the MNIST model versus 4,633,738 in the CIFAR-10 model). 
Algorithm~\ref{A:ModPC} performed similarly to backpropagation in terms of loss and accuracy during learning (Fig.~\ref{F:ModPCCIFAR}A,B) and produced parameter updates that pointed in a similar direction, but still did not match the true gradients (Fig.~\ref{F:ModPCCIFAR}C,D). Algorithm~\ref{A:ModPC} was substantially slower than backpropagation (848s for  Algorithm~\ref{A:ModPC} versus 58s for backpropagation when training metrics were not  computed on every iteration).



 \begin{figure}
\centering{
 \includegraphics[width=4.5in]{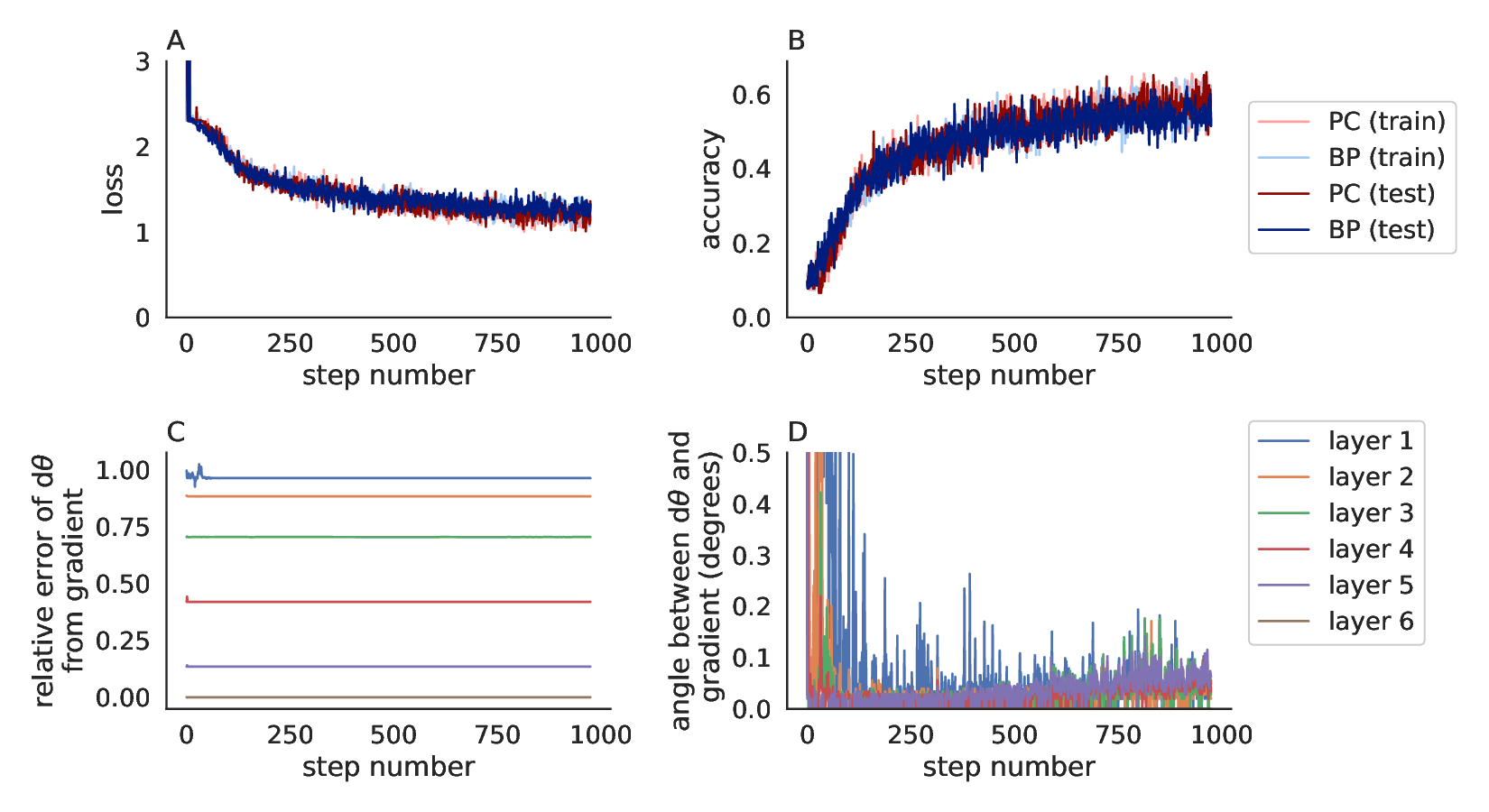}
}
 \caption{{\bf Predictive coding modified by the fixed prediction assumption compared to backpropagation in convolutional neural networks trained on CIFAR-10.} Same as Fig.~\ref{F:ModPC} except a larger network was trained on the CIFAR-10 data set. The accuracy of predictive coding with the fixed prediction assumption is similar to backpropagation and parameter updates are similar to the true gradients.}
 \label{F:ModPCCIFAR}
 \end{figure}

\subsubsection*{Predictive coding modified by the fixed prediction assumption using a step size of $\eta=1$ computes exact gradients in a fixed number of steps.}

A major disadvantage of the approach outlined above -- when compared to standard backpropagation --  is that it requires iterative updates to $v_\ell$ and $\epsilon_\ell$. Indeed, previous work~\cite{millidge2020predictive} used $n=100-200$ iterations, leading to substantially slower performance compared to standard backpropagation. Other work~\cite{whittington2017approximation} used $n=20$ iterations as above. In general, there is a tradeoff between accuracy and performance when choosing $n$, as demonstrated in Fig.~\ref{F:ModPCErrs}. 
However, more recent work~\cite{song2020can,salvatori2021predictive} showed that, under the fixed prediction assumption, predictive coding can compute the exact same gradients computed by backpropagation in a fixed number of steps. That work used a more specific formulation of the neural network which can implement fully connected layers, convolutional layers, and recurrent layers. They also used an unconventional interpretation of neural networks in which weights are multiplied outside the activation function, {\it i.e.}, $f_\ell(x;\theta_\ell)=\theta_\ell g_\ell(x)$, and inputs are fed into the last layer instead of the first. Next, {I} show that their result holds for arbitrary feedforward neural networks as formulated in Eq.~\eqref{E:fwdpass} (with arbitrary functions, $f_\ell$) and this result has a simple interpretation in terms of Algorithm~\ref{A:ModPC}. Specifically, the following theorem shows that taking a step size of $\eta=1$ yields an exact computation of gradients using just $n=L$ iterations (where $L$ is the depth of the network).




\begin{theorem}\label{T:1}
If Algorithm~\ref{A:ModPC} is run with step size $\eta=1$ and at least $n=L$ iterations then the algorithm computes
\[
\epsilon_\ell=\frac{\partial {\mathcal L}(\hat y,y)}{\partial \hat v_\ell}
\]
and
\[
d\theta_\ell=-\frac{\partial {\mathcal L}(\hat y,y)}{\partial \theta_\ell}
\]
for all $\ell=1,\ldots,L$ where $\hat v_\ell=f_\ell(\hat v_{\ell-1};\theta_\ell)$ are the results from a forward pass with $\hat v_0=x$ and  $\hat y=\hat v_L=f(x;\theta)$ is the output.
\end{theorem}
\begin{proof}
For the sake of notational simplicity within this proof, define $\delta_\ell=\partial {\mathcal L}(\hat v_L,y)/\partial \hat v_\ell$. Therefore, we first need to prove that $\epsilon_\ell=\delta_\ell$. 
First, rewrite the inside of the error and belief loop from Algorithm~\ref{A:ModPC}  while explicitly keeping track of the iteration number in which each variable was updated,
\[
\begin{aligned}
\epsilon^i_\ell&=\hat v_\ell-v^{i-1}_\ell\\
dv^i_\ell&=\epsilon^i_\ell-\epsilon^i_{\ell+1}\frac{\partial f_{\ell+1}(\hat v_{\ell};\theta_{\ell+1})}{\partial \hat v_\ell}\\
v^i_\ell&=v^{i-1}_\ell+ dv^i_\ell.
\end{aligned}
\]
Here, $v^i_\ell$, $\epsilon^i_\ell$, and $dv^i_\ell$ denote the values of $v^i_\ell$, $\epsilon^i_\ell$, and $dv^i_\ell$ respectively at the end of the $i$th iteration,  $v^0_\ell=\hat v_\ell$ corresponds to the initial value, and all terms without superscripts are constant inside the inference loop. There are some subtleties here. For example, we have $v_\ell^{i-1}$ in the first line because $v_\ell$ is updated after $\epsilon_\ell$  in the loop. 
More subtly, we have $\epsilon^i_{\ell+1}$ in the second equation instead of $\epsilon^{i-1}_{\ell+1}$ because the for loop goes backwards from $\ell=L-1$ to $\ell=1$, so $\epsilon_{\ell+1}$ is updated before $\epsilon_{\ell}$. 
First note that
\[
\epsilon^1_\ell=0
\]
for $\ell=1,\ldots,L-1$ because $v_\ell^0=\hat v_\ell$. Now compute the change in $\epsilon_\ell$ across one step,
\[
\begin{aligned}
\epsilon^{i}_\ell-\epsilon^{i+1}_\ell&=v^i_\ell-v^{i-1}_\ell\\
&=dv_\ell^i\\
&=\epsilon^i_\ell-\epsilon^i_{\ell+1}\frac{\partial f_{\ell+1}(\hat v_{\ell};\theta_{\ell+1})}{\partial \hat v_\ell}.
\end{aligned}
\]
Note that this equation is only valid for $i\ge 1$ due to the $i-1$ term ($v^{-1}_\ell$ is not defined). 
Adding $\epsilon^i_\ell$ to both sides of the resulting equation gives
\[
\epsilon^{i+1}_\ell=\epsilon^i_{\ell+1}\frac{\partial f_{\ell+1}(\hat v_{\ell};\theta_{\ell+1})}{\partial \hat v_\ell}.
\]
We now use induction to prove that $\epsilon_\ell=\delta_\ell$ after $n=L$ iterations. Indeed, we prove a stronger claim that $\epsilon^i_\ell=\delta_\ell$ at $i=L-\ell+1$. First note that $\epsilon^i_L=\delta_L$ for all $i$ because $\epsilon^i_L$ is initialized to $\delta_L$ and then never changed. Therefore, our claim is true for the base case $\ell=L$.  

Now suppose that $\epsilon^{i}_{\ell+1}=\delta_{\ell+1}$ for $i=L-(\ell+1)+1=L-\ell$. We need to show that $\epsilon^{i+1}_{\ell}=\delta_{\ell}$. From above, we have
\[
\begin{aligned}
\epsilon^{i+1}_{\ell}&=\epsilon^i_{\ell+1}\frac{\partial f_{\ell+1}(\hat v_{\ell};\theta_{\ell+1})}{\partial \hat v_{\ell}}\\
&=\delta_{\ell+1}\frac{\partial f_{\ell+1}(\hat v_{\ell};\theta_{\ell+1})}{\partial \hat v_{\ell}}\\
&=\frac{\partial {\mathcal L}(\hat y,y)}{\partial \hat v_{\ell+1}}\frac{\partial f_{\ell+1}(\hat v_{\ell};\theta_{\ell+1})}{\partial \hat v_{\ell}}\\
&=\frac{\partial {\mathcal L}(\hat y,y)}{\partial \hat v_{\ell+1}}\frac{\partial \hat v_{\ell+1}}{\partial \hat v_{\ell}}\\
&=\frac{\partial {\mathcal L}(\hat y,y)}{\partial \hat v_{\ell}}\\
&=\delta_{\ell}.
\end{aligned}
\]
This completes our induction argument. It follows that $\epsilon^i_\ell=\delta_\ell$ at iteration $i=L-\ell+1$ at all layers $\ell=1,\ldots,L$. The last layer to be updated to the correct value is $\ell=1$, which is updated on iteration number $i=L-1+1=L$. Hence, $\epsilon_\ell=\delta_\ell$ for all $\ell=1,\ldots,L$ after $n=L$ iterations. This proves the first statement in our theorem. The second statement then follows from the definition of $d\theta_\ell$, 
\[
\begin{aligned}
d\theta_\ell&=-\epsilon_\ell\frac{\partial f_\ell(\hat v_{\ell-1};\theta_\ell)}{\partial \theta_\ell}\\
&=-\frac{\partial {\mathcal L}(\hat y,y)}{\partial \hat v_{\ell}}\frac{\partial f_\ell(\hat v_{\ell-1};\theta_\ell)}{\partial \theta_\ell}\\
&=-\frac{\partial {\mathcal L}(\hat y,y)}{\partial \hat v_{\ell}}\frac{\partial \hat v_{\ell}}{\partial \theta_\ell}\\
&=-\frac{\partial {\mathcal L}(\hat y,y)}{\partial \theta_\ell}.
\end{aligned}
\]
This completes the proof.
\end{proof}

This theorem ties together the implementation and formulation of predictive coding from~\cite{millidge2020predictive} ({\it i.e.}, Algorithm~\ref{A:ModPC}) to the results in~\cite{salvatori2021predictive,song2020can}. As noted in~\cite{salvatori2021predictive,song2020can}, this result depends critically on the assumption that the values of $v_\ell$ are initialized to the activations from a forward pass, $v_\ell=\hat v_\ell$ initially. 
The theoretical predictions from Theorem~\ref{T:1} are confirmed by the fact that all of the errors in the rightmost panels of Fig.~\ref{F:ModPCErrs} converge to zero after $n=L=5$ iterations.

 \begin{figure}
\centering{
 \includegraphics[width=5in]{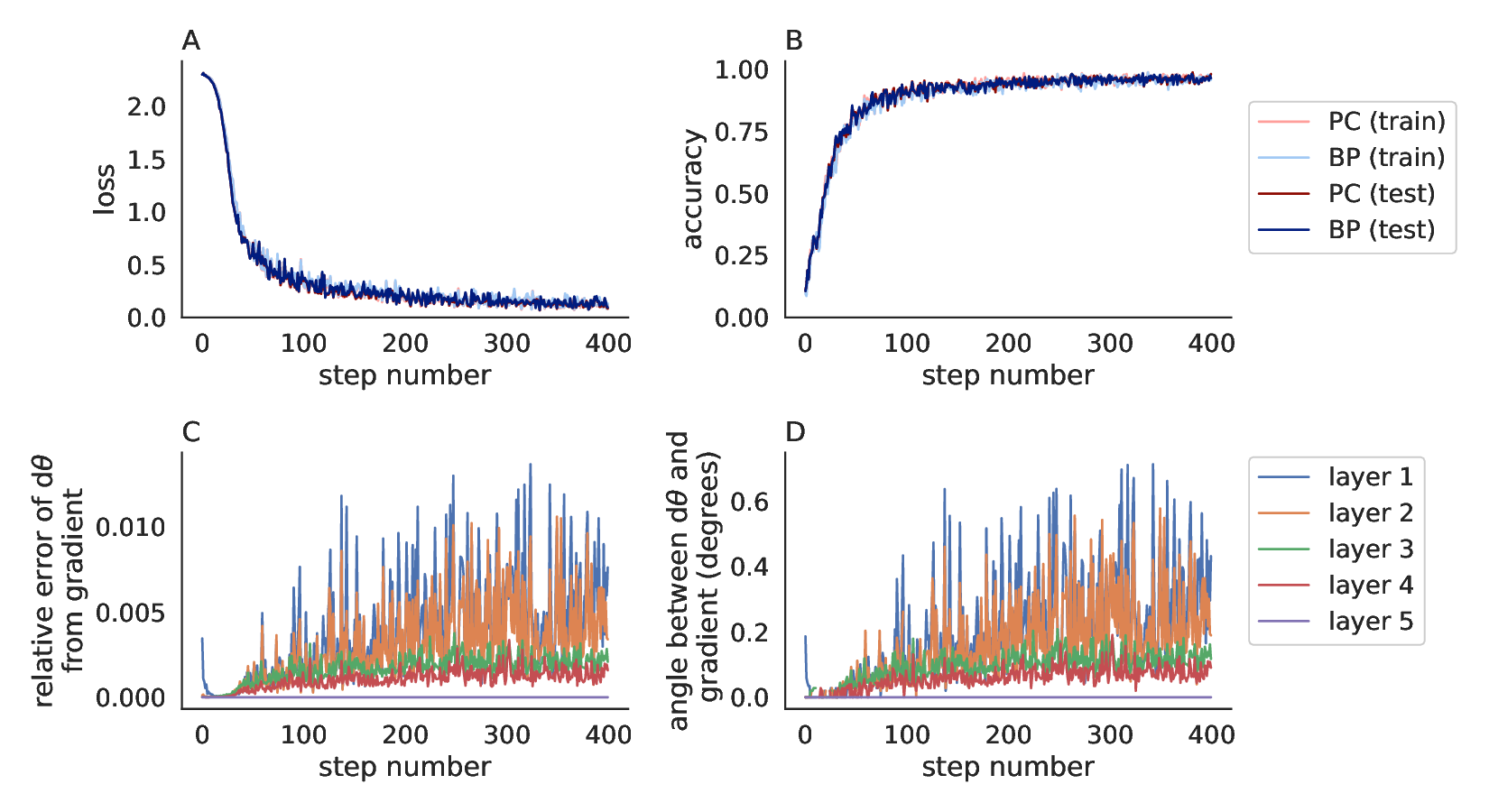}
 }
 \caption{{\bf Predictive coding modified by the fixed prediction assumption with $\eta=1$ compared to backpropagation in convolutional neural networks trained on MNIST.} Same as Fig.~\ref{F:ModPC} except $\eta=1$ and $n=L$.  Predictive coding with the fixed prediction assumption approximates true gradients accurately when $\eta=1$.}
 \label{F:ModPCeta1}
 \end{figure}
 
To further test the result empirically, {I} repeated Figs.~\ref{F:ModPC} and \ref{F:ModPCCIFAR} using $\eta=1$ and $n=L$ (in contrast to Figs.~\ref{F:ModPC}  and \ref{F:ModPCCIFAR} which used $\eta=0.1$ and $n=20$). The loss and accuracy closely matched those computed by backpropagation (Figs.~\ref{F:ModPCeta1}A,B and \ref{F:ModPCeta1CIFAR}A,B). More importantly, the parameter updates closely matched the true gradients (Fig.~\ref{F:ModPCeta1}C,D and \ref{F:ModPCeta1CIFAR}C,D), as predicted by Theorem~\ref{T:1}. 
The differences between predictive coding and backpropagation in Fig.~\ref{F:ModPCeta1} were due floating point errors and the non-determinism of computations performed on GPUs. For example, similar differences to those seen in Fig.~\ref{F:ModPCeta1}A,B were present when the same training algorithm was run twice with the same random seed.  
The smaller number of iterations ($n=L$ in Figs.~\ref{F:ModPCeta1} and \ref{F:ModPCeta1CIFAR}  versus $n=20$ in Figs.~\ref{F:ModPC} and \ref{F:ModPCCIFAR}) resulted in a shorter training time (13s for MNIST and 300s for CIFAR-10 for Figs.~\ref{F:ModPCeta1} and \ref{F:ModPCeta1CIFAR}, compare to 29s and 848s in Figs \ref{F:ModPC} and \ref{F:ModPCCIFAR}, and compare to 8s and 58s for backpropagation). 

In summary, a review of the literature shows that a strict interpretation of predictive coding (Algorithm~\ref{A:PC}) does not converge to the true gradients computed by backpropagation. To compute the true gradients, predictive coding must be modified by the fixed prediction assumption (Algorithm~\ref{A:ModPC}). Further, {I} proved that Algorithm~\ref{A:ModPC} computes the exact gradients when $\eta=1$ and $n\ge L$, which ties together results from previous work~\cite{millidge2020predictive,salvatori2021predictive,song2020can}.

 \begin{figure}
\centering{
 \includegraphics[width=5in]{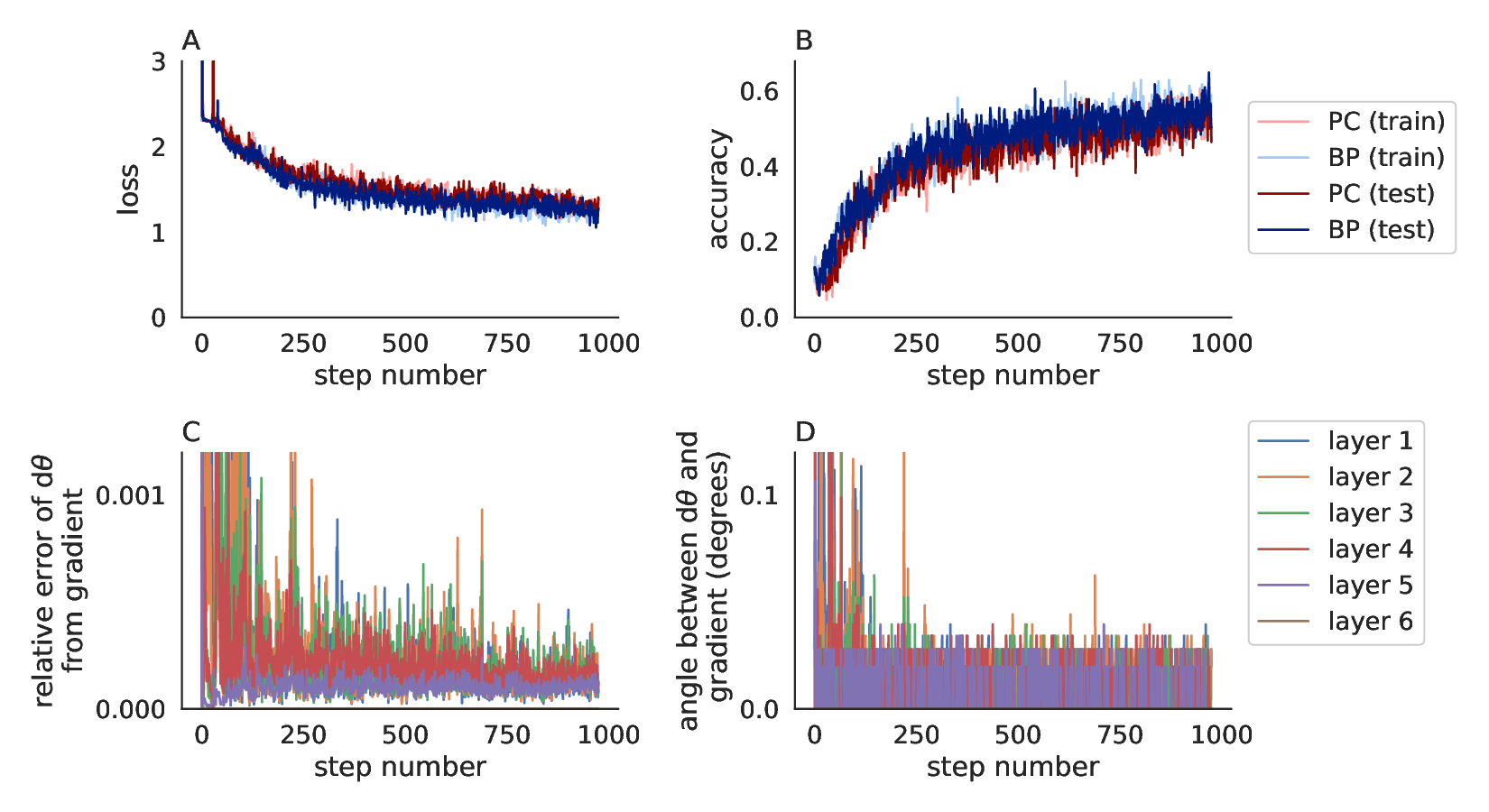}
 }
 \caption{{\bf Predictive coding modified by the fixed prediction assumption with $\eta=1$ compared to backpropagation in convolutional neural networks trained on CIFAR-10.} Same as Fig.~\ref{F:ModPCCIFAR} except $\eta=1$ and $n=L$.  Predictive coding with the fixed prediction assumption approximates true gradients accurately when $\eta=1$.}
 \label{F:ModPCeta1CIFAR}
 \end{figure}


\subsection*{Predictive coding with the fixed prediction assumption and $\eta=1$ is functionally equivalent to a direct implementation of backpropagation.}

The proof of Theorem~\ref{T:1} and the last panel of Fig.~\ref{F:ModPCErrs} give some insight into a how Algorithm~\ref{A:ModPC} works. First note that the values of $v_\ell$ in Algorithm~\ref{A:ModPC} are only used to compute the values of $\epsilon_\ell$ and are not otherwise used in the computation of $d\theta_\ell$ or any other quantities. Therefore, if we only care about understanding parameter updates, $d\theta_\ell$, we can ignore the values of $v_\ell$ and only focus on how $\epsilon_\ell$ is updated on each iteration, $i$. Secondly, note that when $\eta=1$, each $\epsilon_\ell$ is updated only once: $\epsilon^i_\ell=0$ for $i<L-\ell+1$ and $\epsilon^i_\ell=\epsilon^i_{\ell+1}{\partial f_{\ell+1}(\hat v_{\ell};\theta_{\ell+1})}/{\partial \hat v_{\ell}}$ for $i\ge L-\ell+1$, so $\epsilon_\ell$ is only changed on iteration number $i=L-\ell+1$. In other words, the error computation in Algorithm~\ref{A:ModPC} when $\eta=1$ and $n=L$ is equivalent to
\begin{algorithmic}
 \vspace{.04in}
\State {\color{gray}\# error computation}
\State $\epsilon_L=\tfrac{\partial {\mathcal L}(\hat v_L,y)}{\partial  \hat v_L}$
\For{$i=1,\ldots,L$}
	\For{$\ell=L-1,\ldots,1$}
		\If{$\ell==L-i+1$}
			\State $\epsilon_\ell=\epsilon_{\ell+1}\tfrac{\partial f_{\ell+1}(\hat v_{\ell};\theta_{\ell+1})}{\partial \hat v_\ell}$
		\EndIf
	\EndFor
\EndFor
\end{algorithmic}
The two computations are equivalent in the sense that they compute the same values of the errors, $\epsilon_\ell^i$, on every iteration. The formulation above makes it clear that the nested loops are unnecessary because for each value of $i$, $\epsilon_\ell$ is only updated at one value of $\ell$. Therefore, the nested loops and if-statement can be replaced by a single for-loop. Specifically, the error computation  in Algorithm~\ref{A:ModPC} when $\eta=1$ is equivalent to
\begin{algorithmic}
 \vspace{.04in}
\State {\color{gray}\# error computation}
\State $\epsilon_L=\tfrac{\partial {\mathcal L}(\hat v_L,y)}{\partial  \hat v_L}$
\For{$\ell=L-1,\ldots,1$}
	\State $\epsilon_\ell=\epsilon_{\ell+1}\tfrac{\partial f_{\ell+1}(\hat v_{\ell};\theta_{\ell+1})}{\partial \hat v_\ell}$
\EndFor
\end{algorithmic}
This is {\it exactly} the error computation from the standard backpropagation algorithm, {\it i.e.}, Algorithm~\ref{A:backprop}. Hence, if we use $\eta=1$, then Algorithm~\ref{A:ModPC} is just  backpropagation with extra steps and these extra steps do not compute any non-zero values. If we additionally want to compute the fixed point beliefs, then they can still be computed using the relationship
\[
v_\ell=\hat v_\ell-\epsilon_\ell.
\]
We may conclude that, when $\eta=1$, Algorithm~\ref{A:ModPC} can be replaced by an exact implementation of backpropagation without any effect on the results or effective implementation of the algorithm. This raises the question of whether predictive coding with the fixed prediction assumption should be considered any more biologically plausible than a direct implementation of backpropagation.


\subsection*{Accounting for covariance or precision matrices in hidden layers does not affect learning under the fixed prediction assumption.}

Above, {I} showed that predictive coding with the fixed prediction assumption is functionally equivalent to backpropagation. However, the predictive coding algorithm was derived under an assumption that covariance matrices in the probabilistic  model are identity matrices, $\Sigma_\ell=I$. This raises the question of whether relaxing this assumption could generalize backpropagation to account for the covariances, as suggested in previous work~\cite{whittington2017approximation,millidge2020predictive,millidge2021predictive}. 

We can account for covariances by returning to the calculations starting from the probabilistic model in Eq~\eqref{E:ProbModel}  and omit the assumption that $\Sigma_\ell=I$. To this end, it is helpful to define the precision-weighted prediction errors~\cite{bogacz2017tutorial,buckley2017free,millidge2021predictive}, 
\[
\tilde \epsilon_\ell=\epsilon_\ell P_\ell
\]
for $\ell=1,\ldots,L-1$ where $P_\ell=\Sigma_\ell^{-1}$ is the inverse of the covariance matrix of $V_\ell$, which is called ``precision matrix.''  Recall that we treat $\epsilon_\ell$ as a row-matrix, which explains the right-multiplication in this definition. 

Modifying  the definition of $\epsilon_L$ to account for covariances is not so simple because the Gaussian model for $V_\ell$ is not justified for non-Euclidean loss functions such as categorical loss functions. Moreover, it is not clear how to define the covariance or precision matrix of the output layer when labels are observed. As such, {I} restrict to accounting for precision matrices in hidden layers only, and leave the question of accounting for covariances in the output layer for future work with some comments on the issue provided at the end of this section. 
To this end, let us not modify the last layer's precision and instead define 
\[
\tilde \epsilon_L=\epsilon_L=\frac{\partial L(\hat y,y)}{\partial \hat y}.
\] 
The free energy is then defined as~\cite{bogacz2017tutorial,buckley2017free}
\[
F=\frac{1}{2}\sum_{\ell=1}^L \|\tilde \epsilon_\ell\|^2.
\]
Performing gradient descent on $F$ with respect to $v_\ell$ therefore gives
\[
dv_\ell= \tilde \epsilon_\ell - \tilde \epsilon_{\ell+1}\frac{\partial f_{\ell+1}(v_\ell;\theta_{\ell+1})}{\partial v_\ell}
\]
and performing gradient descent on $F$ with respect to $\theta_\ell$ gives 
\[
d\theta_\ell= -\tilde \epsilon_\ell \frac{\partial f_{\ell}(v_{\ell-1};\theta_{\ell})}{\partial \theta_\ell}.
\]
These expressions are identical to Eqs.~\eqref{E:dvell} and \eqref{E:dtheta} derived above except that $\tilde \epsilon_\ell$ takes the place of $\epsilon_\ell$. 

The precision matrices themselves can be learned by performing gradient descent on $F$ with respect to $P_\ell$ or, as suggested in other work~\cite{bogacz2017tutorial}, by parameterizing the model in terms of $\Sigma_\ell=P_\ell^{-1}$ and performing gradient descent with respect to $\Sigma_\ell$.   Alternatively, one could use techniques from the literature on Gaussian graphical models to learn a sparse or low-rank representation of $P_\ell$.  {I} circumvent the question of estimating $P_\ell$ by instead just asking how an estimate of $P_\ell$ (however it is obtained) would affect learning. {I} do assume that $P_\ell$ is symmetric. {I} also simplify the calculations by restricting the analysis to predictive coding with the fixed prediction assumption, leaving the analysis of fixed point prediction errors and parameter updates under strict predictive coding with precisions matrices for future work. Some analysis has been performed in this direction~\cite{bogacz2017tutorial}, but not for the supervised learning scenario considered here. 

Putting this together, predictive coding under the fixed prediction assumption while accounting for precision matrices in hidden layers is defined by the following equations
\[
\begin{aligned}
\tilde \epsilon_\ell&=\left[\hat v_\ell-v_\ell\right]P_\ell\\
\tilde \epsilon_L&=\frac{\partial {\mathcal L}(\hat v_L,y)}{\partial \hat v_L}\\
dv_\ell&=\tilde \epsilon_\ell-\tilde \epsilon_{\ell+1}\frac{\partial f_{\ell+1}(\hat v_{\ell};\theta_{\ell+1})}{\partial \hat v_\ell}\\
d\theta_\ell&=-\tilde \epsilon_\ell\frac{\partial f_\ell(\hat v_{\ell-1};\theta_\ell)}{\partial \theta_\ell}
\end{aligned}
\]
The only difference between these equations and Eqs.~\eqref{E:ModPC} is that they use $\tilde \epsilon_\ell$ in place of $\epsilon_\ell=\hat v_\ell-v_\ell$. Following the same line of reasoning, therefore, if the updates to $v_\ell$ are repeated until convergence, then fixed point precision-weighted prediction errors satisfy
\[
\tilde \epsilon^*_\ell=\tilde \epsilon^*_{\ell+1}\frac{\partial f_{\ell+1}(\hat v_{\ell};\theta_{\ell+1})}{\partial \hat v_\ell}.
\]
Notably, this is the same equation derived for $\epsilon_\ell$ under the fixed prediction assumption with $\Sigma_\ell=I$, so fixed point precision-weighted prediction errors are also the same,
\[
\tilde \epsilon^*_\ell=\frac{\partial {\mathcal L}(\hat y,y)}{\partial \hat v_{\ell}}
\]
and, therefore, parameter updates are the same as well,
\[
d\theta_\ell=-\frac{\partial {\mathcal L}(\hat y,y)}{\partial \theta_\ell}.
\]
In conclusion, accounting for precision matrices in hidden layers does not affect learning under the fixed prediction assumption. Fixed point parameter updates are still the same as those computed by backpropagation.  This conclusion is independent of how the precision matrices are estimated, but it does rely on the assumption that fixed points for $v_\ell$ exist and are unique. 

Above, we only considered precision matrices in the hidden layers because accounting for precision matrices in the output layer is problematic for general loss functions. The use of a precision matrix in the output implies the use of a Gaussian model for the output layer and labels, which is inconsistent with some types of  labels and loss functions. If we focus on the case of a squared-Euclidean loss function,
\[
{\mathcal L}(\hat y,y)=\frac{1}{2}\|\hat y-y\|^2,
\]
then the use of precision  matrices in the output layer is more parsimonious and we can define 
\[
\tilde \epsilon_L=\left[\hat v_L-y\right]P_L=\frac{\partial \mathcal L(\hat y,y)}{\partial \hat y}P_L
\]
 in place of the definition above (recalling that $\hat y=\hat v_L$). Following the same calculations as above, gives fixed points of the form
\[
\tilde \epsilon^*_\ell=\frac{\partial \mathcal L(\hat y,y)}{\partial \hat y}P_L\frac{\partial \hat y}{\partial \hat v_\ell}
\]
and, therefore, weight updates take the form 
\[
d\theta_\ell=-\frac{\partial \mathcal L(\hat y,y) }{\partial \hat y}P_L\frac{\partial \hat y}{\partial \theta_\ell}.
\]
at the fixed point. Hence, accounting for precision matrices at the output layer can affect learning by re-weighting the gradient of the loss function according to the precision matrix of the output layer. Note that  the precision matrices of the hidden layers still have no effect on learning in this case.   Previous work relates the inclusion of the precision matrix in output layers with the use of natural gradients~\cite{amari1995information,millidge2021predictive}.


\subsection*{Prediction errors do not necessarily represent surprising or unexpected features of inputs.}

Deep neural networks are often interpreted as abstract models of cortical neuronal networks. To this end, the activations of units in deep neural networks are compared to the activity (typically firing rates) of cortical neurons~\cite{SchrimpfKubilius2018BrainScore,Schrimpf2020integrative,lillicrap2020backpropagation}. This approach ignores the representation of errors within the network. More generally, the activations in one particular layer of a feedforward deep neural network contain no information about the activations of deeper layers, the label, or the loss. On the other hand, the activity of cortical neurons can be modulated by downstream activity and information believed to be passed upstream by feedback projections. Predictive coding provides a precise model for the information that deeper layers send to shallower layers, specifically prediction errors.

Under the fixed prediction assumption (Algorithm~\ref{A:ModPC}), prediction errors in a particular layer are  approximated by the gradients of that layers' activations with respect to the loss function, $\epsilon_\ell= \delta_\ell=\frac{\partial {\mathcal L}}{\partial \hat v_\ell}$, but under a strict interpretation of predictive coding (Algorithm~\ref{A:PC}), prediction errors do not necessarily reflect gradients. 
We next empirically explored how the representations of images differ between the activations from a feedforward pass, $\hat v_\ell$,  the prediction errors under the fixed prediction assumption, $\epsilon_\ell= \delta_\ell$, as well as the beliefs, $v_\ell$, and prediction errors, $\epsilon_\ell$, under a strict interpretation of predictive coding (Algorithm~\ref{A:PC}). To do so, we computed each quantity in VGG-19~\cite{simonyan2014very}, which is a large, feedforward convolutional neural network (19 layers and 143,667,240 trainable parameters) pre-trained on ImageNet~\cite{russakovsky2015imagenet}.

The use of convolutional layers allowed us to visualize the activations and prediction errors in each layer. Specifically, we took the Euclidean norm of each quantity across all channels and plotted them as two-dimensional images for layers $\ell=1$ and $\ell=10$ and for two different input images (Fig.~\ref{F:VGG19}). For each image and each layer (each row in Fig.~\ref{F:VGG19}), we computed the Euclidean norm of four quantities. First, we computed the activations from a forward pass through the network ($\hat v_\ell$, second column). Under predictive coding with the fixed prediction assumption (Algorithm~\ref{A:ModPC}), we can interpret the activations, $\hat v_\ell$, as ``beliefs'' and the gradients, $\delta_\ell$, as ``prediction errors.'' Strictly speaking, there is a distinction between the beliefs, $\hat v_\ell$, from a feedforward pass and the beliefs, $v_\ell=\hat v_\ell+\epsilon_\ell$, when labels are provided. Either could be interpreted as a ``belief.'' However, we found that the difference between them was negligible  for the examples  considered here. 

Next, we computed the gradients of the loss with respect to the activations ($\delta_\ell$, third column in Fig.~\ref{F:VGG19}). The theory and simulations above and from previous work confirms that these gradients approximate the prediction errors from predictive coding with the fixed prediction assumption (Algorithm~\ref{A:ModPC}). Indeed, for the examples considered here, the differences between the two quantities were negligible. Next, we computed the beliefs ($v_\ell$, fourth column in Fig.~\ref{F:VGG19}) computed by strict predictive coding (Algorithm~\ref{A:PC}). Finally, we computed the prediction errors ($\epsilon_\ell$, last column in Fig.~\ref{F:VGG19}) computed by strict predictive coding (Algorithm~\ref{A:PC}).

Note that we used a VGG-19 model that was pre-trained using backpropagation. Hence, the weights were not necessarily the same as the weights that would be obtained if the model were trained using predictive coding, particularly strict predictive coding (Algorithm~\ref{A:PC}) which does not necessarily converge to the true gradients. Training a large ImageNet model like VGG-19 with predictive coding is extremely computationally expensive. Regardless, future work should address the question of whether using pre-trained weights (versus weights trained by predictive coding) affects the conclusions reached here.

\begin{figure}
\centering{
 \includegraphics[width=5in]{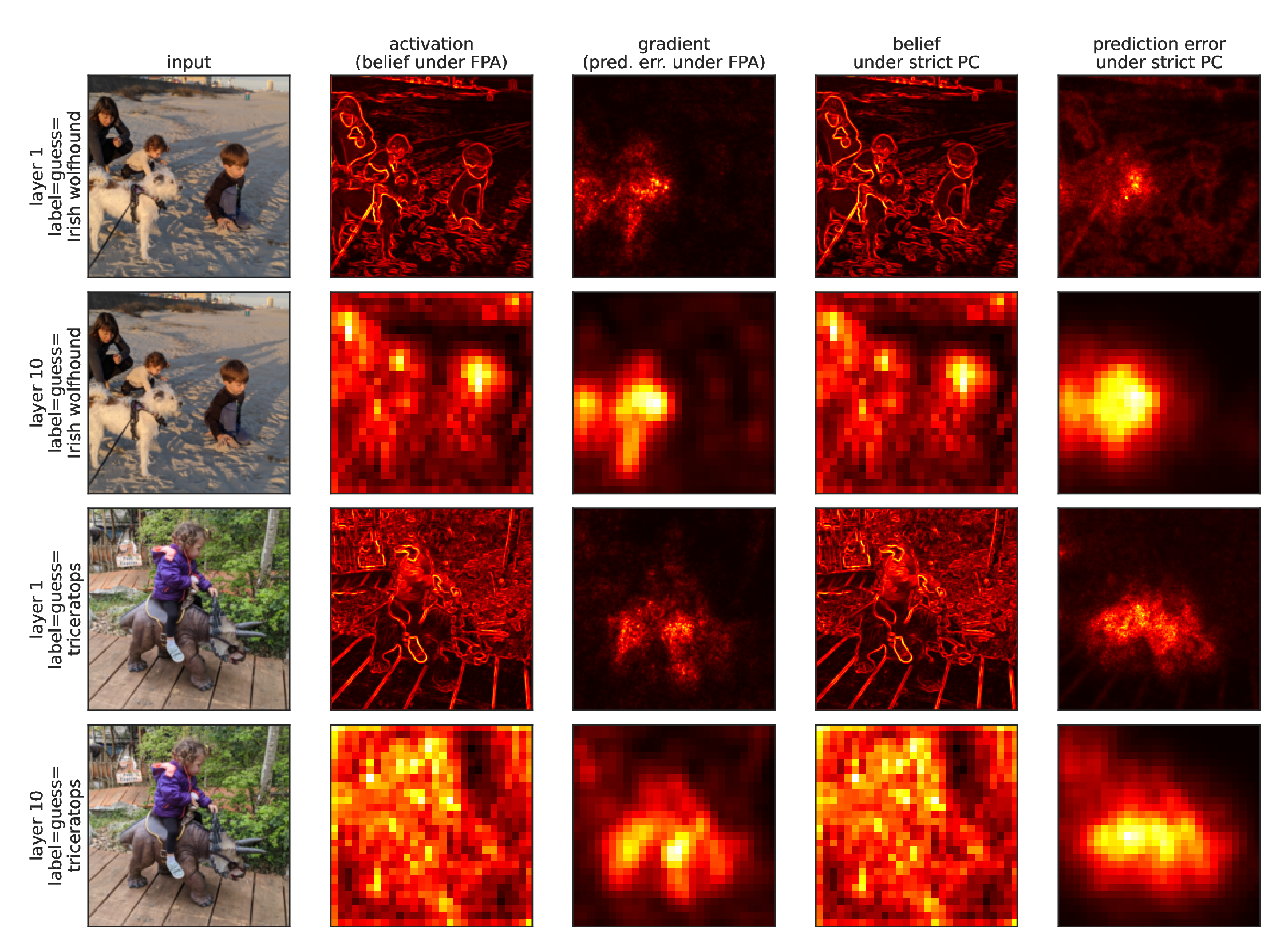}
 }
 \caption{{\bf Magnitude of activations, beliefs, and prediction errors in a convolutional neural network pre-trained on ImageNet.} The Euclidean norm of feedforward activations ($\hat v$, interpreted as beliefs  under the fixed prediction assumption), gradients of the loss with respect to activations ($\delta_\ell=\partial {\mathcal L}/\partial \hat v$, interpreted as prediction errors  under the fixed prediction assumption), beliefs ($v$) under strict predictive coding, and prediction errors ($\epsilon_\ell)$) under strict predictive coding  computed from the VGG-19 network~\cite{simonyan2014very} pre-trained on ImageNet~\cite{russakovsky2015imagenet} with  two different photographs as inputs at two different layers. The vertical labels on the left (``triceratops'' and ``Irish wolfhound'') correspond to the guessed label which was also used as  the ``true'' label ($y$) used to compute the gradients.   }
 \label{F:VGG19}
 \end{figure}

Overall, the activations, $\hat v_\ell$, from a feedforward pass were qualitatively very similar to the beliefs, $v_\ell$, computed under a strict interpretation of predictive coding (Algorithm~\ref{A:PC}). To a slightly lesser degree, the gradients, $\delta_\ell$, from a feedforward pass were qualitatively similar to the prediction errors computed under a strict interpretation of predictive coding (Algorithm~\ref{A:PC}). Since $\hat v_\ell$ and $\delta_\ell$ approximate beliefs and prediction errors under the fixed prediction assumption, these observations confirmed that the fixed prediction assumption does not make large qualitative changes to the representation of beliefs and errors in these examples. Therefore, in the discussion below, we  used ``beliefs'' and ``prediction errors'' to refer to the quantities from both models.

Interestingly, prediction errors were non-zero even when the image and the network's ``guess'' was consistent with the label (no ``mismatch''). Indeed, the prediction errors were largest in magnitude at pixels corresponding to the object predicted by the label, {\it i.e.}, at the most predictable regions. While this observation is an obvious consequence of the fact that prediction errors are approximated by the gradients, $\delta_\ell=\frac{\partial {\mathcal L}}{\partial \hat v_\ell}$, it is contradictory to the heuristic or intuitive interpretation of prediction errors as measurements of ``surprise'' in the colloquial sense of the word~\cite{friston2010free}. 

As an illustrative example from Fig.~\ref{F:VGG19}, it is not surprising that an image labeled by ``triceratops'' contains a triceratops, but this does not imply a lack of prediction errors because the space of images containing a triceratops is large and any one image of a triceratops is not wholly representative of the label. Moreover, the pixels to which the loss is most sensitive are those pixels containing the triceratops. Therefore those pixels give rise to larger values of $\epsilon_\ell\approx \delta_\ell=\partial \mathcal L/\partial \hat v_\ell$. Hence, in high-dimensional sensory spaces, predictive coding models do not necessarily predict that prediction error units encode ``surprise'' in the colloquial sense of the word.

In both examples in Fig.~\ref{F:VGG19}, we used an input, $y$, that matched the network's ``guessed'' label, {\it i.e.}, the label to which the network assigned the highest probability ($\argmax(\hat y)$). Prediction errors are often discussed in the context of mismatched stimuli in which top-down input is inconsistent with bottom-up predictions~\cite{hertag2020learning,gillon2021learning,keller2012sensorimotor,zmarz2016mismatch,attinger2017visuomotor,homann2017predictive}. Mismatches can be modeled by  taking a label that is different from the network's guess. In Fig.~\ref{F:VGG19MM}, we visualized the prediction errors in response to matched and mismatched labels. The network assigned a probability of $p=0.9991$ to the label ``carousel'' and a probability of $p=3.63\times 10^{-8}$ to the label ``bald eagle''. The low probability assigned to ``bald eagle'' is, at least in part, a consequence of the network being trained with a softmax loss function, which implicitly assumes one label per input.
When we applied the mismatched label ``bald eagle,'' prediction errors were larger in pixels that are salient for that label ({\it e.g.}, the bird's white head, which is  a defining feature of a bald eagle). Moreover, the prediction errors as a whole are much larger in magnitude in response to the mismatched label (see the scales of the color bars in Fig.~\ref{F:VGG19MM}).

\begin{figure}
\centering{
\includegraphics[width=5in]{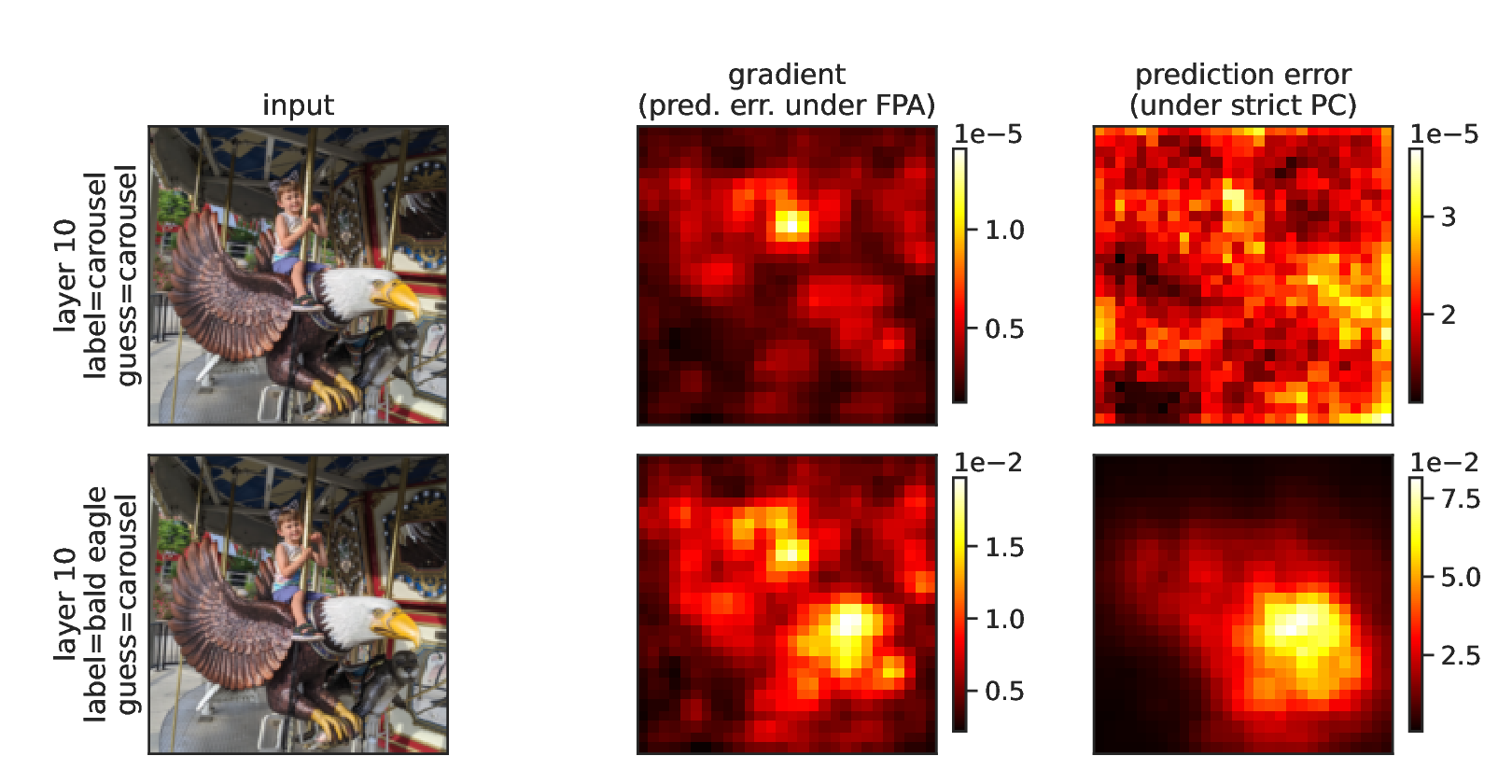}
 }
 \caption{{\bf Magnitude of activations, beliefs, and prediction errors in response to matched and mismatched inputs and labels.} Same as Fig.~\ref{F:VGG19}, but for the bottom row the label did not match the network's guess.  }
 \label{F:VGG19MM}
 \end{figure}



In summary, the relationship between prediction errors and gradients helped demonstrate that  prediction errors sometimes, but do not always conform to their common interpretation as unexpected features of a bottom-up input in the context of a top-down input. Also, beliefs and prediction errors were qualitatively similar with and without the fixed prediction assumption for the examples considered here.

\section*{Discussion}

We reviewed and extended previous work~\cite{ whittington2017approximation,millidge2020predictive,song2020can,salvatori2021predictive} on the relationship between predictive coding and backpropagation for learning in neural networks. Our results demonstrated that a strict interpretation of predictive coding does not accurately approximate backpropagation, but is still capable of learning (Figs.~\ref{F:PC} and \ref{F:PCErrs}). 
Previous work  proposed a modification to predictive coding called the ``fixed prediction assumption'' which causes predictive coding to converge to the same parameter updates produced by backpropagation, under the assumption that the predictive coding iterations converge to fixed points. 
Hence, the relationship between predictive coding and backpropagation identified in previous work relies critically on the fixed prediction assumption. Formal derivations of predictive coding in terms of  variational inference~\cite{buckley2017free} do not produce the fixed prediction assumption. It is possible that an alternative probabilistic model or alternative approaches to the variational formulation could help formalize a model of predictive coding under the fixed prediction assumption. 

We proved analytically and verified empirically that taking a step size of $\eta=1$ in the modified predictive coding algorithm computes the exact gradients computed by backpropagation in a fixed number of steps (modulo floating point numerical errors). This result is consistent with similar, but slightly less general, results in previous work~\cite{salvatori2021predictive,song2020can}. 

A closer inspection of the the fixed prediction assumption with $\eta=1$ showed that it is algorithmically equivalent to a direct implementation of backpropagation. As such, any potential neural architecture and machinery that could be to implement predictive coding with the fixed prediction assumption could also implement backpropagation directly. This result calls into question whether predictive coding with the fixed prediction assumption is any more biologically plausible than a direct implementation of backpropagation.

Visualizing the beliefs and prediction errors produced by predictive coding models applied to a large convolutional neural network pre-trained on ImageNet showed that beliefs and prediction errors were activated by distinct parts of input images, and the parts of the images that produced larger prediction errors were not always consistent with an intuitive interpretation of prediction errors as representing surprising or unexpected features of inputs. 
These observations are consistent with the fact that prediction errors approximate gradients of the loss function in backpropagation~\cite{whittington2017approximation,millidge2020predictive,song2020can,salvatori2021predictive}. Gradients are large for input features that have a larger impact on the loss. While surprising features can have a large impact on the loss, unsurprising features can as well. We only verified this finding empirically on a few examples. The reader can try additional examples by inserting the URL of any image into the file \verb#PredErrsFromURLimage.ipynb# contained in the directories linked in Materials and Methods, and can also be accessed directly at  
\url{https://bit.ly/3JwGUM9}. 
Future work should attempt to quantify the relationship between prediction errors and surprising features more systematically across many inputs. In addition, prediction errors could be computed for learning tasks associated with common experimental paradigms so they can be used to make experimentally testable predictions.

When interpreting artificial deep neural networks as models of biological neuronal networks, it is common to compare activations in the artificial network to biological neurons' firing rates~\cite{SchrimpfKubilius2018BrainScore,Schrimpf2020integrative}. However, under predictive coding models and other models in which errors are propagated upstream by feedback connections, many biological interpretations posit the existence of ``error neurons'' that encode the errors sent upstream. In most such models (including predictive coding), error neurons reflect or approximate the gradient of the loss function with respect to artificial neurons' activations, $\delta_\ell$. Any model that hypothesizes the neural representation of backpropagated errors would predict that some recorded neural activity should reflect these errors. Therefore, if we want to draw analogues between artificial and biological neural networks, the activity of biological neurons should be compared  to both the activations {\it and} the gradients of artificial neurons.

Following previous work~\cite{millidge2020predictive,whittington2017approximation}, we took the covariance matrices underlying the probabilistic model to be identity matrices, $\Sigma_\ell=I$, when deriving the predictive coding model. We also showed that relaxing this assumption by allowing for arbitrary precision matrices in hidden layers  does not affect learning under the fixed prediction assumption. 
Future work should consider the utility of accounting for covariance (or precision) matrices in models without the fixed prediction assumption ({\it i.e.}, under the ``strict'' model) and accounting for precisions or covariances in  the output layer. Moreover, precision matrices could still have benefits in other settings such as recurrent network models, unsupervised learning, or active inference.

Predictive coding and deep neural networks (trained by backpropagation) are often viewed as competing models of brain function. 
Better understanding their relationship can help in the interpretation and implementation of each algorithm as well as their mutual relationships to biological neuronal networks. 

\section*{Materials and Methods}

All numerical examples were performed on GPUs using Google Collaboratory with custom-written PyTorch code. 
The networks trained on MNIST used two convolutional and three fully connected layers with rectified linear activation functions using 2 epochs, a learning rate of $0.002$, and a batch size of 300. The networks trained on CIFAR-10 used three convolutional and three fully connected layers with rectified linear activation functions using 5 epochs, a learning rate of $0.01$, and a batch size of 256. All networks were trained using the Adam optimizer with gradients replaced by the output of the respective algorithm. 
A Google Drive folder with Colab notebooks that produce all figures in this text can be found at\\ 
\noindent 
\url{https://drive.google.com/drive/folders/1m_y0G_sTF-pV9pd2_sysWt1nvRvHYzX0}\\ 
\noindent A copy of the same code is also stored at\\ 
\noindent \url{https://github.com/RobertRosenbaum/PredictiveCodingVsBackProp}\\
\noindent Full details of the neural network architectures and metaparameters can be found in this code. 

\subsection*{Torch2PC software for predictive coding with PyTorch models}

The figures above were all produced using PyTorch~\cite{paszke2019pytorch} models combined with custom written functions for predictive coding. 
Functions for predictive coding with PyTorch models are collected in the Github Repository Torch2PC.
 Currently, the only available functions are intended for models built using the Sequential class, but more general functions will be added to \texttt{Torch2PC} in the future. The functions can be imported using the following commands
\begin{verbatim}
!git clone https://github.com/RobertRosenbaum/Torch2PC.git
from Torch2PC import TorchSeq2PC as T2PC
\end{verbatim}
The primary function in \texttt{TorchSeq2PC} is \texttt{PCInfer}, 
which performs one predictive coding step (computes one value of $d\theta$) on a batch of inputs and labels. The function takes an input \texttt{ErrType}, which is a string that determines whether to use a strict interpretation of predictive coding (Algorithm~2; \texttt{ErrType="Strict"}), predictive coding with the fixed prediction assumption (Algorithm~3; \texttt{"FixedPred"}), or to compute the gradients exactly using backpropagation (Algorithm~1; \texttt{"Exact"}). 
Algorithm~2 can be called as follows,
\begin{verbatim}
vhat,Loss,dLdy,v,epsilon=
  T2PC.PCInfer(model,LossFun,X,Y,"Strict",eta,n,vinit)
\end{verbatim}
where \texttt{model} is a Sequential PyTorch model, \texttt{LossFun} is a loss function, \texttt{X} is a mini-batch of inputs, \texttt{Y} is a mini-batch of labels, \texttt{eta} is the step size, \texttt{n} is the number of iterations to use, and \texttt{vinit} is the initial value for the beliefs. If \texttt{vinit} is not passed, it is set to the result from a forward pass, \texttt{vinit=vhat}. The function returns a list of activations from a forward pass at each layer as \texttt{vhat}, the loss as \texttt{Loss}, the gradient of the output with respect to the loss as \texttt{dLdy}, a list of beliefs, $v_\ell$, at each layer as \texttt{v}, and a list of prediction errors, $\epsilon_\ell$, at each layer as \texttt{epsilon}. The values of the parameter updates, $d\theta_\ell$, are stored in the \texttt{grad} attributes of each parameter, \texttt{model.param.grad}. Hence, after a call to \texttt{PCInfer}, gradient descent could be implemented by calling
\begin{verbatim}
with torch.no_grad():
  for p in modelPC.parameters():
    p-=eta*p.grad
\end{verbatim}
Alternatively, an arbitrary optimizer could be used by calling
\begin{verbatim}
optimizer.step() 
\end{verbatim}
where \texttt{optimizer} is an optimizer created using the PyTorch \texttt{optim} class, {\it e.g.}, by calling \\ \texttt{optimizer=optim.Adam(model.parameters())} before the call to \texttt{T2PC.PCInfer}.

The input \texttt{model} should be a PyTorch Sequential model. Each layer is treated as a single predictive coding layer. Multiple functions can be included within the same layer by wrapping them in a separate call to \texttt{Sequential}. For example the following code:
\begin{verbatim}
model=nn.Sequential(    
    nn.Conv2d(1,10,3),
    nn.ReLU(),
    nn.MaxPool2d(2),
    nn.Conv2d(10,10,3),
    nn.ReLU())
\end{verbatim}
will treat each item as its own layer (5 layers in all). To treat each "convolutional block" as a separate layer, instead do
\begin{verbatim}
model=nn.Sequential(
    
    nn.Sequential(
      nn.Conv2d(1,10,3),
      nn.ReLU(),
      nn.MaxPool2d(2)),
    
    nn.Sequential(
      nn.Conv2d(10,10,3),
      nn.ReLU()))
\end{verbatim}
which has just 2 layers.

Algorithm~3 can be called as follows,
\begin{verbatim}
vhat,Loss,dLdy,v,epsilon=
  T2PC.PCInfer(model,LossFun,X,Y,"FixedPred",eta,n)
\end{verbatim}
The input \texttt{vinit} is not used for Algorithm~3, so it does not need to be passed in. 
The exact values computed by backpropagation can be obtained by calling 
\begin{verbatim}
vhat,Loss,dLdy,v,epsilon=
  T2PC.PCInfer(model,LossFun,X,Y,"Exact")
\end{verbatim}
The inputs \texttt{vinit}, \texttt{eta}, and \texttt{n} are not used for computing exact gradients, so they do not need to be passed in. Theorem~1 says that 
\begin{verbatim}
T2PC.PCInfer(model,LossFun,X,Y,"FixedPred",eta=1,n=len(model))
\end{verbatim}
computes the same values as 
\begin{verbatim}
T2PC.PCInfer(model,LossFun,X,Y,"Exact")
\end{verbatim}
up to numerical floating point errors. 
The inputs \texttt{eta}, \texttt{n}, and \texttt{vinit} are optional. If they are omitted by calling
\begin{verbatim}
T2PC.PCInfer(model,LossFun,X,Y,ErrType)
\end{verbatim}
then they default to \texttt{eta=.1,n=20,vinit=None} which produces \texttt{vinit=vhat} when \\
\texttt{ErrType="Strict"}. 
More complete documentation and a complete example is provided as \\
 \verb#SimpleExample.ipynb# in the GitHub repository and in the code accompanying this paper. More examples are provided by the code accompanying each figure above. 

\clearpage

\section*{Supporting information}

\paragraph*{S1 Figure.}\label{S:1}
~\\~

\includegraphics[width=4.5in]{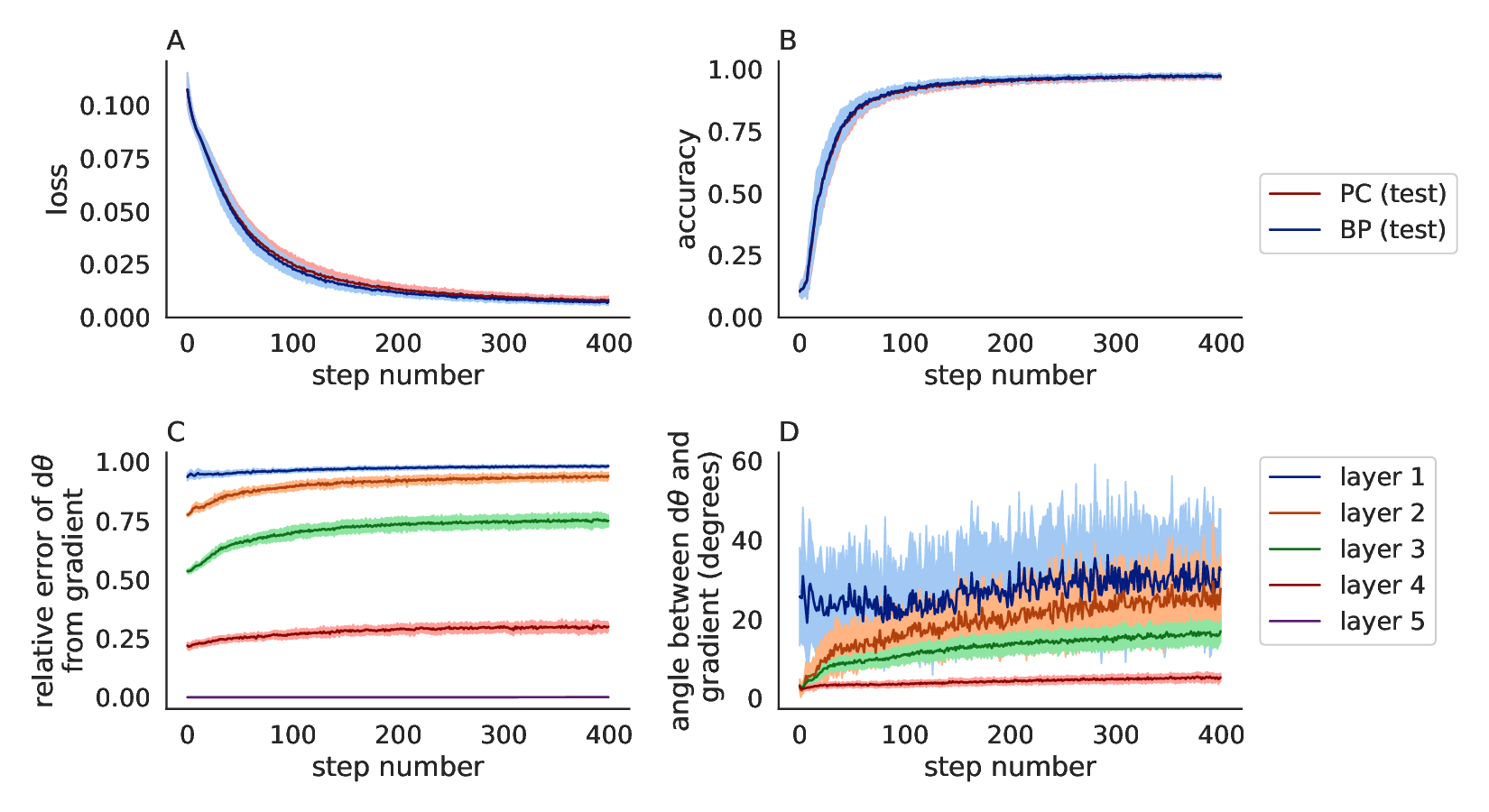}
 \noindent {{\bf Comparing backpropagation and predictive coding in a convolutional neural network trained on MNIST across multiple trials.} 
Same as Figure~\ref{F:PC} except the model was trained 30 times with different random seeds. Dark curves show the mean values and shaded regions show $\pm$ one standard deviation across trials. 
}

\clearpage

\paragraph*{S2 Figure.}\label{S:2}
~\\~

\includegraphics[width=4.5in]{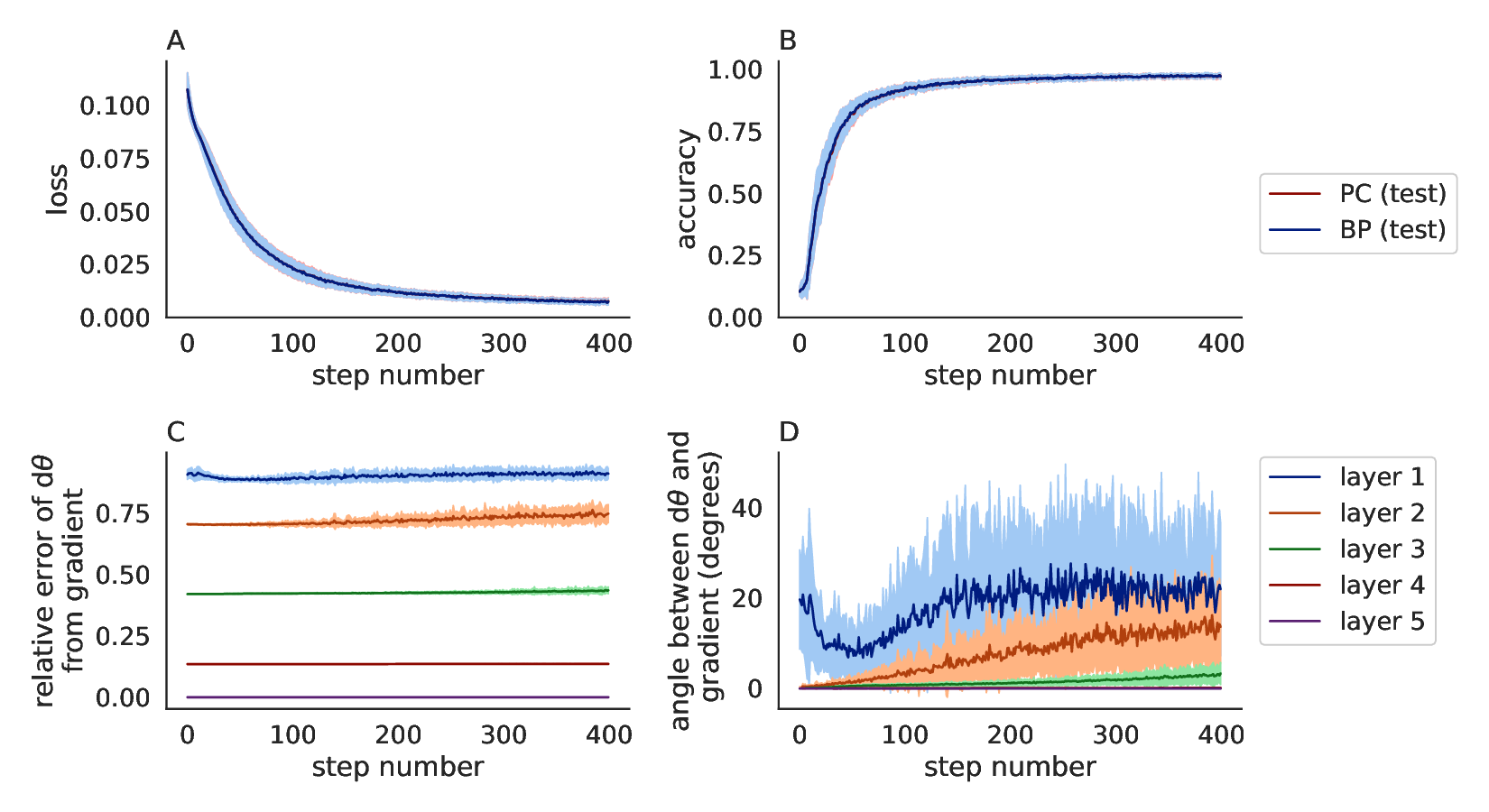}
 \noindent {{\bf Comparing backpropagation and predictive coding modified by the fixed prediction assumption in a convolutional neural network trained on MNIST across multiple trials.} 
Same as Figure~\ref{F:ModPC} except the model was trained 30 times with different random seeds. Dark curves show the mean values and shaded regions show $\pm$ one standard deviation across trials. 
}

\clearpage


\begin{thebibliography}{10}

\bibitem{izhikevich2007solving}
Izhikevich EM.
\newblock Solving the distal reward problem through linkage of STDP and
  dopamine signaling.
\newblock Cerebral cortex. 2007;17(10):2443--2452.

\bibitem{clark2021credit}
Clark DG, Abbott L, Chung S.
\newblock Credit Assignment Through Broadcasting a Global Error Vector.
\newblock arXiv preprint arXiv:210604089. 2021;.

\bibitem{lillicrap2020backpropagation}
Lillicrap TP, Santoro A, Marris L, Akerman CJ, Hinton G.
\newblock Backpropagation and the brain.
\newblock Nature Reviews Neuroscience. 2020;21(6):335--346.

\bibitem{whittington2019theories}
Whittington JC, Bogacz R.
\newblock Theories of error back-propagation in the brain.
\newblock Trends in Cognitive Sciences. 2019;23(3):235--250.

\bibitem{urbanczik2014learning}
Urbanczik R, Senn W.
\newblock Learning by the dendritic prediction of somatic spiking.
\newblock Neuron. 2014;81(3):521--528.

\bibitem{lillicrap2016random}
Lillicrap TP, Cownden D, Tweed DB, Akerman CJ.
\newblock Random synaptic feedback weights support error backpropagation for
  deep learning.
\newblock Nature Communications. 2016;7(1):1--10.

\bibitem{scellier2017equilibrium}
Scellier B, Bengio Y.
\newblock Equilibrium propagation: Bridging the gap between energy-based models
  and backpropagation.
\newblock Frontiers in computational neuroscience. 2017;11:24.

\bibitem{aljadeff2019cortical}
Aljadeff J, D'amour J, Field RE, Froemke RC, Clopath C.
\newblock Cortical credit assignment by Hebbian, neuromodulatory and inhibitory
  plasticity.
\newblock arXiv preprint arXiv:191100307. 2019;.

\bibitem{kunin2020two}
Kunin D, Nayebi A, Sagastuy-Brena J, Ganguli S, Bloom J, Yamins D.
\newblock Two routes to scalable credit assignment without weight symmetry.
\newblock In: International Conference on Machine Learning. PMLR; 2020. p.
  5511--5521.

\bibitem{payeur2021burst}
Payeur A, Guerguiev J, Zenke F, Richards BA, Naud R.
\newblock Burst-dependent synaptic plasticity can coordinate learning in
  hierarchical circuits.
\newblock Nature Neuroscience. 2021; p. 1--10.

\bibitem{whittington2017approximation}
Whittington JC, Bogacz R.
\newblock An approximation of the error backpropagation algorithm in a
  predictive coding network with local hebbian synaptic plasticity.
\newblock Neural Computation. 2017;29(5):1229--1262.

\bibitem{millidge2020predictive}
Millidge B, Tschantz A, Buckley CL.
\newblock Predictive coding approximates backprop along arbitrary computation
  graphs.
\newblock arXiv preprint arXiv:200604182. 2020;.

\bibitem{song2020can}
Song Y, Lukasiewicz T, Xu Z, Bogacz R.
\newblock Can the brain do backpropagation?—exact implementation of
  backpropagation in predictive coding networks.
\newblock Advances in Neural Information Processing Systems. 2020;33:22566.

\bibitem{salvatori2021predictive}
Salvatori T, Song Y, Lukasiewicz T, Bogacz R, Xu Z.
\newblock Predictive Coding Can Do Exact Backpropagation on Convolutional and
  Recurrent Neural Networks.
\newblock arXiv preprint arXiv:210303725. 2021;.

\bibitem{rao1999predictive}
Rao RP, Ballard DH.
\newblock Predictive coding in the visual cortex: a functional interpretation
  of some extra-classical receptive-field effects.
\newblock Nature Neuroscience. 1999;2(1):79--87.

\bibitem{friston2010free}
Friston K.
\newblock The free-energy principle: a unified brain theory?
\newblock Nature Reviews Neuroscience. 2010;11(2):127--138.

\bibitem{huang2011predictive}
Huang Y, Rao RP.
\newblock Predictive Coding.
\newblock Wiley Interdisciplinary Reviews: Cognitive Science.
  2011;2(5):580--593.

\bibitem{bastos2012canonical}
Bastos AM, Usrey WM, Adams RA, Mangun GR, Fries P, Friston KJ.
\newblock Canonical microcircuits for predictive coding.
\newblock Neuron. 2012;76(4):695--711.

\bibitem{clark2015surfing}
Clark A.
\newblock Surfing uncertainty: Prediction, action, and the embodied mind.
\newblock Oxford University Press; 2015.

\bibitem{buckley2017free}
Buckley CL, Kim CS, McGregor S, Seth AK.
\newblock The free energy principle for action and perception: A mathematical
  review.
\newblock Journal of Mathematical Psychology. 2017;81:55--79.

\bibitem{bogacz2017tutorial}
Bogacz R.
\newblock A tutorial on the free-energy framework for modelling perception and
  learning.
\newblock Journal of Mathematical Psychology. 2017;76:198--211.

\bibitem{spratling2017review}
Spratling MW.
\newblock A review of predictive coding algorithms.
\newblock Brain and cognition. 2017;112:92--97.

\bibitem{keller2018predictive}
Keller GB, Mrsic-Flogel TD.
\newblock Predictive processing: a canonical cortical computation.
\newblock Neuron. 2018;100(2):424--435.

\bibitem{goodfellow2016deep}
Goodfellow I, Bengio Y, Courville A, Bengio Y.
\newblock Deep Learning.
\newblock MIT press Cambridge; 2016.

\bibitem{krizhevsky2009learning}
Krizhevsky A, Hinton G, et~al.
\newblock Learning multiple layers of features from tiny images.
\newblock Citeseer. 2009;.

\bibitem{millidge2021predictive}
Millidge B, Seth A, Buckley CL.
\newblock Predictive Coding: a Theoretical and Experimental Review.
\newblock arXiv preprint arXiv:210712979. 2021;.

\bibitem{amari1995information}
Amari Si.
\newblock Information geometry of the EM and em algorithms for neural networks.
\newblock Neural networks. 1995;8(9):1379--1408.

\bibitem{SchrimpfKubilius2018BrainScore}
Schrimpf M, Kubilius J, Hong H, Majaj NJ, Rajalingham R, Issa EB, et~al.
\newblock Brain-Score: Which Artificial Neural Network for Object Recognition
  is most Brain-Like?
\newblock bioRxiv preprint. 2018;.

\bibitem{Schrimpf2020integrative}
Schrimpf M, Kubilius J, Lee MJ, Murty NAR, Ajemian R, DiCarlo JJ.
\newblock Integrative Benchmarking to Advance Neurally Mechanistic Models of
  Human Intelligence.
\newblock Neuron. 2020;.

\bibitem{simonyan2014very}
Simonyan K, Zisserman A.
\newblock Very deep convolutional networks for large-scale image recognition.
\newblock arXiv preprint arXiv:14091556. 2014;.

\bibitem{russakovsky2015imagenet}
Russakovsky O, Deng J, Su H, Krause J, Satheesh S, Ma S, et~al.
\newblock Imagenet large scale visual recognition challenge.
\newblock International Journal of Computer Vision. 2015;115(3):211--252.

\bibitem{hertag2020learning}
Hert{\"a}g L, Sprekeler H.
\newblock Learning prediction error neurons in a canonical interneuron circuit.
\newblock Elife. 2020;9:e57541.

\bibitem{gillon2021learning}
Gillon CJ, Pina JE, Lecoq JA, Ahmed R, Billeh Y, Caldejon S, et~al.
\newblock Learning from unexpected events in the neocortical microcircuit.
\newblock bioRxiv. 2021;.

\bibitem{keller2012sensorimotor}
Keller GB, Bonhoeffer T, H{\"u}bener M.
\newblock Sensorimotor mismatch signals in primary visual cortex of the
  behaving mouse.
\newblock Neuron. 2012;74(5):809--815.

\bibitem{zmarz2016mismatch}
Zmarz P, Keller GB.
\newblock Mismatch receptive fields in mouse visual cortex.
\newblock Neuron. 2016;92(4):766--772.

\bibitem{attinger2017visuomotor}
Attinger A, Wang B, Keller GB.
\newblock Visuomotor coupling shapes the functional development of mouse visual
  cortex.
\newblock Cell. 2017;169(7):1291--1302.

\bibitem{homann2017predictive}
Homann J, Koay SA, Glidden AM, Tank DW, Berry MJ.
\newblock Predictive coding of novel versus familiar stimuli in the primary
  visual cortex.
\newblock BioRxiv. 2017; p. 197608.

\bibitem{paszke2019pytorch}
Paszke A, Gross S, Massa F, Lerer A, Bradbury J, Chanan G, et~al.
\newblock PyTorch: An Imperative Style, High-Performance Deep Learning Library.
\newblock Advances in Neural Information Processing Systems.
  2019;32:8026--8037.

\end{thebibliography}

\end{document}